\documentclass[letterpaper,twocolumn,10pt]{article}
\usepackage{usenix-2020-09}

\usepackage{tikz}
\usepackage{amsmath}
\usepackage{xspace}
\usepackage[normalem]{ulem}
\usepackage{bm}

\usepackage{cleveref}
\crefname{algocf}{alg.}{algs.}
\Crefname{algocf}{Alg.}{Algs.}
\Crefname{equation}{Eq.}{Eqs.}
\Crefname{figure}{Fig.}{Figs.}
\Crefname{table}{Tab.}{Tabs.}
\Crefname{appendix}{Appx.}{Appxs.}
\Crefname{theorem}{Thm.}{Thms.}

\usepackage{amsmath}
\usepackage{listings}
\usepackage{multirow}
\usepackage{tabularx}
\usepackage{caption}
\usepackage{booktabs}
\usepackage{pifont}
\usepackage[super]{nth}
\usepackage{footnote}
\usepackage{tablefootnote}
\usepackage{multicol}
\usepackage{textcomp}
\usepackage{xcolor}
\usepackage{soul}
\usepackage{xfrac}
\usepackage{enumitem}
\usepackage{textcomp}
\usepackage{subcaption}

\setlist[itemize]{noitemsep}
\setlist[enumerate]{noitemsep}

\newboolean{compress}
\setboolean{compress}{true}
\setlist[itemize]{leftmargin=*}
\setlist[enumerate]{leftmargin=*}
\setlist{itemsep=0pt,parsep=0pt} 

\usepackage[subtle,wordspacing=normal]{savetrees} 

\makesavenoteenv{tabular}
\makesavenoteenv{table}

\usepackage[font={small,bf,sf}, textfont={small,sf}]{caption}
\usepackage{float}
\usepackage[small,compact]{titlesec}

\newlength{\gapspace}

\definecolor{PrpCoral}{HTML}{FF7F50}

\definecolor{AlexBlue}{HTML}{0048BA}

\definecolor{BurcuPurple}{HTML}{330066}

\usepackage{tikz}
\newcommand*\myc[1]{%
\scalebox{0.78}{\begin{tikzpicture}[baseline=-4pt]
  \node[draw,circle,inner sep=0.5pt, fill=black] {\textcolor{white}{\textsf{\textbf{#1}}}};
\end{tikzpicture}}}

\newcommand{\tinyskip}{\vspace{1pt}}

\newcommand{\mypar}[1]{\tinyskip\noindent\textbf{#1.}\xspace}
\newcommand{\myparr}[1]{\tinyskip\noindent\textbf{#1}\xspace}

\newcommand{\eg}{e.g.,\xspace}
\newcommand{\ie}{i.e.,\xspace}
\newcommand{\etal}{et al.\xspace}

\newcommand{\unit}[1]{\mbox{\hspace{2pt}#1}\xspace}

\newcommand{\sys}{IA-CCF\xspace}
\newcommand{\sysfull}{Individual Accountability for CCF\xspace}

\newcommand{\LPBFT}{L-PBFT\xspace}
\newcommand{\LPBFTfull}{Ledger Practical Byzantine Fault Tolerance\xspace}
\newcommand{\msg}[1]{\textsf{#1}\xspace}

\definecolor{darkcyan}{rgb}{0.0, 0.55, 0.55}

\usepackage[linesnumbered,ruled,nofillcomment,norelsize]{algorithm2e}
\SetKwRepeat{Do}{do}{while}%
\SetKwFor{While}{while}{}{end while}%
\SetKwBlock{KwFunction}{function}{}
\SetKwBlock{KwOn}{on}{}
\SetKwBlock{KwRepeat}{repeat}{}
\SetKwBlock{KwPeriodically}{periodically}{}
\SetKwBlock{KwState}{state}{}
\SetKwBlock{KwAtomic}{atomically}{}
\SetKwBlock{KwConcurrently}{concurrently}{}
\SetCommentSty{textup} 
\SetKwComment{KwIttayComment}{(}{)} 
\SetKw{KwAnd}{and} 
\SetKw{KwOr}{or}
\SetKw{KwBreak}{Break}
\SetKw{KwNot}{not}
\SetKw{KwSetTimer}{setTimer}
\SetKwBlock{KwExpTimer}{on timer expiration}{}
\SetKwFor{ForAll}{forall}{}

\SetKwFunction{FMain}{Main}

\SetKwFunction{FVerifyReceipt}{\textsf{\footnotesize verifyReceipt}}
\SetKwFunction{FCreateReceipt}{\textsf{\footnotesize createReceipt}}
\SetKwFunction{FReplay}{\textsf{\footnotesize replay}}
\SetKwFunction{FLoadCheckpoint}{\textsf{\footnotesize loadCheckpoint}}
\SetKwFunction{FReplay}{\textsf{\footnotesize replay}}
\SetKwFunction{FVerifyBatch}{\textsf{\footnotesize verifyBatch}}
\SetKwFunction{FVerifyInBatch}{\textsf{\footnotesize verifyInBatch}}
\SetKwFunction{FExecuteBatch}{\textsf{\footnotesize executeBatch}}
\SetKwFunction{FBlame}{\textsf{\footnotesize blame}}
\SetKwFunction{FPathHash}{\textsf{\footnotesize pathHash}}
\SetKwFunction{FHash}{\textsf{\footnotesize hash}}
\SetKwFunction{FVerify}{\textsf{\footnotesize verify}}
\SetKwFunction{FLoad}{\textsf{\footnotesize load}}
\SetKwFunction{FValidSignatures}{\textsf{\footnotesize validSignatures}}
\SetKwFunction{FGetMerkleLeaf}{\textsf{\footnotesize getMerkleLeaf}}
\SetKwFunction{FGetSisterBranch}{\textsf{\footnotesize getSisterBranch}}
\SetKwFunction{FGoToParent}{\textsf{\footnotesize goToParent}}
\SetKwFunction{FNewReceipt}{\textsf{\footnotesize newReceipt}}
\SetKwFunction{FCheckSignature}{\textsf{\footnotesize checkSignature}}
\SetKwFunction{FCheckSignatureAndHash}{\textsf{\footnotesize checkSignatureAndHash}}
\SetKwFunction{FMostCommon}{\textsf{\footnotesize mostCommon}}
\SetKwFunction{FMin}{\textsf{\footnotesize min}}
\SetKwFunction{FInc}{\textsf{\footnotesize inc}}
\SetKwFunction{FDec}{\textsf{\footnotesize dec}}
\SetKwFunction{FIsPrimary}{\textsf{\footnotesize isPrimary}}
\SetKwFunction{FIsBackup}{\textsf{\footnotesize isBackup}}
\SetKwFunction{FBecomePrimary}{\textsf{\footnotesize becomePrimary}}
\SetKwFunction{FAddPending}{\textsf{\footnotesize addPending}}
\SetKwFunction{FCreatePrePrepare}{\textsf{\footnotesize createPrePrepare}}
\SetKwFunction{FCreatePrepare}{\textsf{\footnotesize createPrepare}}
\SetKwFunction{FCreateCommit}{\textsf{\footnotesize createCommit}}
\SetKwFunction{FCreateViewChange}{\textsf{\footnotesize createViewChange}}
\SetKwFunction{FCreateNewView}{\textsf{\footnotesize createNewView}}
\SetKwFunction{FExecuteBatch}{\textsf{\footnotesize executeBatch}}
\SetKwFunction{FVerifyExecute}{\textsf{\footnotesize verifyExecution}}
\SetKwFunction{FGetPendingRequests}{\textsf{\footnotesize getPendingReqs}}
\SetKwFunction{FVerifyViewChanges}{\textsf{\footnotesize verifyViewChanges}}
\SetKwFunction{FVerifyPrePrepare}{\textsf{\footnotesize verifyPrePrepare}}
\SetKwFunction{FVerifyPrepare}{\textsf{\footnotesize verifyPrepare}}
\SetKwFunction{FVerifyCommit}{\textsf{\footnotesize verifyCommit}}
\SetKwFunction{FVerifyViewChange}{\textsf{\footnotesize verifyViewChange}}
\SetKwFunction{FVerifyNewView}{\textsf{\footnotesize verifyNewView}}
\SetKwFunction{FSendNewView}{\textsf{\footnotesize sendNewView}}
\SetKwFunction{FReceiveRequest}{\textsf{\footnotesize receiveTransactionRequest}}
\SetKwFunction{FReceivePrePrepare}{\textsf{\footnotesize receivePrePrepare}}
\SetKwFunction{FReceivePrepare}{\textsf{\footnotesize receivePrepare}}
\SetKwFunction{FReceiveCommit}{\textsf{\footnotesize receiveCommit}}
\SetKwFunction{FReceiveViewChange}{\textsf{\footnotesize receiveViewChange}}
\SetKwFunction{FSendViewChange}{\textsf{\footnotesize sendViewChange}}
\SetKwFunction{FReceiveNewView}{\textsf{\footnotesize receiveNewView}}
\SetKwFunction{FSendToAllReplicas}{\textsf{\footnotesize sendToAllReplicas}}
\SetKwFunction{FSendViewChangeMessage}{\textsf{\footnotesize sendViewChangeMessage}}
\SetKwFunction{FBatchPrepared}{\textsf{\footnotesize batchPrepared}}
\SetKwFunction{FBatchCommitted}{\textsf{\footnotesize batchCommitted}}
\SetKwFunction{FIsBatchPrepared}{\textsf{\footnotesize isBatchPrepared}}
\SetKwFunction{FIsBatchCommitted}{\textsf{\footnotesize isBatchCommitted}}
\SetKwFunction{FCreateNewBatch}{\textsf{\footnotesize createNewBatch}}
\SetKwFunction{FSendReplyToClient}{\textsf{\footnotesize sendReplyToClient}}
\SetKwFunction{FSendReceiptToClient}{\textsf{\footnotesize sendReceiptToClient}}
\SetKwFunction{FRevertToLastCommitted}{\textsf{\footnotesize revertToLastCommitted}}
\SetKwFunction{FNeedMoreEvidence}{\textsf{\footnotesize needMoreEvidence}}
\SetKwFunction{FSetView}{\textsf{\footnotesize setView}}
\SetKwFunction{FResetLedger}{\textsf{\footnotesize resetLedger}}
\SetKwFunction{FIsGovernance}{\textsf{\footnotesize isGovernance}}
\SetKwFunction{FIsCheckpoint}{\textsf{\footnotesize isCheckpoint}}
\SetKwFunction{FCreateCheckpoint}{\textsf{\footnotesize createCheckpoint}}
\SetKwFunction{FCreateNonce}{\textsf{\footnotesize createNonce}}
\SetKwFunction{FAppendToMerkleTree}{\textsf{\footnotesize appendToMerkleTree}}
\SetKwFunction{FAppendToLedger}{\textsf{\footnotesize appendToLedger}}
\SetKwFunction{FMaxPrepared}{\textsf{\footnotesize maxPrepared}}
\SetKwFunction{FMaxComitted}{\textsf{\footnotesize maxCommitted}}
\SetKwFunction{FNewViewReady}{\textsf{\footnotesize newViewReady}}
\SetKwFunction{FClearAfter}{\textsf{\footnotesize clearAfter}}
\SetKwFunction{FGetTxForBatch}{\textsf{\footnotesize getTxForBatch}}
\SetKwFunction{FShouldSendReceipt}{\textsf{\footnotesize shouldSendReceipt}}
\SetKwFunction{FExecute}{\textsf{\footnotesize execute}}
\SetKwFunction{FGetTx}{\textsf{\footnotesize getTx}}
\SetKwFunction{FGetPrepares}{\textsf{\footnotesize getPrepares}}
\SetKwFunction{FGetPrePrepares}{\textsf{\footnotesize getPrePrepares}}
\SetKwFunction{FGetCommits}{\textsf{\footnotesize getCommits}}
\SetKwFunction{FGetNonces}{\textsf{\footnotesize getNonces}}
\SetKwFunction{FHasPrePrepare}{\textsf{\footnotesize hasPrePrepare}}
\SetKwFunction{FGetPrePrepare}{\textsf{\footnotesize getPrePrepare}}
\SetKwFunction{FGetEvidence}{\textsf{\footnotesize getEvidence}}
\SetKwFunction{FHasEvidence}{\textsf{\footnotesize hasEvidence}}
\SetKwFunction{FGetMerkleRoot}{\textsf{\footnotesize getRoot}}
\SetKwFunction{FGetMerklePath}{\textsf{\footnotesize getMerklePath}}
\SetKwFunction{FCreateMerklePaths}{\textsf{\footnotesize createMerklePaths}}
\SetKwFunction{FHasPrepares}{\textsf{\footnotesize hasPrepares}}
\SetKwFunction{FGetViewChanges}{\textsf{\footnotesize getViewChanges}}
\SetKwFunction{FSelectViewWithHighestPP}{\textsf{\footnotesize selectViewWithHighestPP}}
\SetKwFunction{FFetchAndApplyMissingLedgerEntries}{\textsf{\footnotesize fetchAndApplyMissingLedgerEntries}}
\SetKwFunction{FCreateEvidenceBitMap}{\textsf{\footnotesize createEvidenceBitMap}}
\SetKwFunction{FUndoNewView}{\textsf{\footnotesize undoNewView}}
\SetKwFunction{FVerifyReplay}{\textsf{\footnotesize verifyReplay}}
\SetKwFunction{FSendPrePrepare}{\textsf{\footnotesize sendPrePrepare}}
\SetKwFunction{FAudit}{\textsf{\footnotesize audit}}
\SetKwFunction{FAuditReceipts}{\textsf{\footnotesize auditReceipts}}
\SetKwFunction{FGetCheckpointAndLedger}{\textsf{\footnotesize getCheckpointAndLedger}}
\SetKwFunction{FVerifyReceiptsInLedger}{\textsf{\footnotesize verifyReceiptsInLedger}}
\SetKwFunction{FReplayLedger}{\textsf{\footnotesize replayLedger}}
\SetKwFunction{FCreateLedgerFragment}{\textsf{\footnotesize createLedgerFragment}}
\SetKwFunction{FIsBatchWellformed}{\textsf{\footnotesize isBatchWellformed}}
\SetKwFunction{FEnforcerGetLedgerPackage}{\textsf{\footnotesize enforcerGetLedgerPackage}}
\SetKwFunction{FGetCheckpointAndLedger}{\textsf{\footnotesize getCheckpointAndLedger}}

\SetKwProg{Fn}{fn}{}{}
\SetKwInput{KwPre}{Pre}

\usepackage{amsthm,lipsum}
\newtheoremstyle{theorem-small}
{3pt}
{3pt}
{\itshape}
{}
{\bfseries}
{.}
{.5em}
{}
\theoremstyle{theorem-small}
\newtheorem{theorem-small}{Theorem}[section]

\newtheorem{lemma}{Lemma}
\newtheorem{theorem}{Theorem}

\newenvironment{proof-sketch}{\noindent\textit{Proof (sketch).}}{\hfill}

\newcommand{\setmuskip}[2]{#1=#2\relax}
\setmuskip{\medmuskip}{0mu} 

\SetAlCapNameFnt{\small\bfseries\sffamily}
\SetAlCapFnt{\small\bfseries\sffamily}
\SetAlgorithmName{Alg.}{Alg.}{List of Algorithms}

\usepackage{authblk}
\makeatletter
\renewcommand\AB@affilsepx{, \protect\Affilfont}
\makeatother

\begin{document}

\date{}

\title{\Large \bf \sys: Individual Accountability for Permissioned Ledgers}

\author[1,2]{Alex Shamis}
\author[1,2]{Peter Pietzuch}
\author[3]{Burcu Canakci\thanks{Work done while at Microsoft Research.}\ \ }
\author[1]{Miguel Castro}
\author[1]{C\'edric Fournet}
\author[1]{Edward Ashton}
\author[1]{Amaury Chamayou}
\author[1]{Sylvan Clebsch}
\author[1]{Antoine Delignat-Lavaud}
\author[4]{Matthew Kerner}
\author[1]{Julien Maffre}
\author[1]{Olga Vrousgou}
\author[1]{Christoph M. Wintersteiger}
\author[1]{Manuel Costa}
\author[4]{Mark Russinovich}
\affil[1]{Microsoft Research}
\affil[2]{Imperial College London}
\affil[3]{Cornell University}
\affil[4]{Microsoft Azure}

\maketitle


\begin{abstract}
  Permissioned ledger systems allow a consortium of members that do not trust one another to execute transactions safely on a set of replicas. Such systems typically use Byzantine fault tolerance~(BFT) protocols to distribute trust, which only ensures safety when fewer than 1/3 of the replicas misbehave. Providing guarantees beyond this threshold is a challenge: current systems assume that the ledger is corrupt and fail to identify misbehaving replicas or hold the members that operate them accountable---instead all members share the blame.

  We describe \sys, a new permissioned ledger system that provides \emph{individual accountability}. It can assign blame to the individual members that operate
  misbehaving replicas regardless of the number of misbehaving replicas or members.
  \sys achieves this by signing and logging BFT protocol messages in the ledger, and by using Merkle trees to provide clients with succinct, universally-verifiable \emph{receipts} as evidence of successful transaction execution. Anyone can \emph{audit} the ledger against a set of receipts to discover inconsistencies and identify replicas that signed contradictory statements. \sys also supports \emph{changes} to consortium membership and replicas by tracking signing keys using a sub-ledger of governance transactions. \sys provides strong disincentives to misbehavior with low overhead: it executes 47,000\unit{tx/s} while providing clients with receipts in two network round trips.
\end{abstract}



\section{Introduction}
\label{intro}

Permissioned ledger systems, such as Hyperledger Fabric~\cite{hyperledger_fabric}, Quorum~\cite{Quorum} and Diem~\cite{libra}, allow a consortium of members that do not trust one another to deploy a trustworthy service on a set of replicas that they operate. These systems typically use protocols for Byzantine fault tolerant~(BFT) state machine replication~\cite{pbft, sbft, hotstuff, bessani2014state, kotla2007zyzzyva, aardvark} to distribute trust: clients send requests to execute transactions~\cite{smart_contracts, wood2014ethereum} that are executed in a consistent order by the replicas. The results are recorded in a persistent, replicated ledger. 

BFT protocols ensure safety (linearizability~\cite{linearizability}) and liveness, but they can only do this if fewer than $1/3$ of $N$~replicas misbehave. With more misbehaving replicas, current permissioned ledger systems can no longer be trusted. When safety violations are detected, 
the whole service is deemed to have failed, and all members and replicas share the blame. 

Current systems try to avoid this problem by increasing replication~\cite{byzcoin, sbft, hotstuff}
or hardening individual replicas~\cite{russinovich2019ccf}. Adding replicas does not help if they are controlled by the same consortium members and thus do not behave independently. Increasing the number of consortium members, however, is challenging or even infeasible in practice. For example, the Diem Association~\cite{diem_association} had 26~members, which prevented it from offering a service with more than 26~independent replicas; other consortia are smaller, which results in fewer independent replicas~\cite{evidence1, evidence2, evidence3}. Even for large consortia with reputable companies, a persistent attacker may slowly compromise $N/3$~replicas over time, \eg by exploiting lax security practices, bribing members' employees or exploiting software vulnerabilities. Without accountability after a service compromise, there is also no perceived reputational loss that would incentivize members to prevent or disclose these incidents~\cite{evidence4, evidence5, evidence6}.

The Confidential Consortium Framework (CCF)~\cite{russinovich2019ccf} uses trusted hardware~\cite{costan2016intel, kaplan2016amd} to isolate replicas from operators and members, and it provides receipts that commit transaction execution to its ledger. However, CCF does not offer safety or individual accountability if the trusted hardware is compromised.

Prior work explores accountability for various types of distributed systems~\cite{yumerefendi2005role, peerreview, avm, bar, leibowitz2019no}. PeerReview~\cite{peerreview} makes general message passing systems accountable. As we show in \S\ref{sec:eval}, applying such a general approach to a permissioned ledger system incurs high overhead: all messages must be signed, and auditing is expensive, because it correlates logs across many replicas. More recent work~\cite{kannan2020bft, polygraph, ranchal2020blockchain, buchman2016tendermint} investigates accountability in BFT protocols and blockchains. These proposals, however, offer no guarantees when $2/3$ or more replicas misbehave, because misbehaving replicas may rewrite the ledger history without detection.

We describe \emph{\sysfull}~(\sys), a BFT permissioned ledger system that identifies misbehaving replicas and assigns blame to the 
individual members that operate them, \emph{even if all replicas misbehave}. Individual accountability provides strong disincentives for misbehavior. 

\sys is a prototype that extends CCF~\cite{russinovich2019ccf} with support for BFT and individual accountability, while retaining the same user programming model, key-value store, transaction execution engine, and model of governance for changes to the consortium membership and replica set.

\sys supports individual accountability by introducing \emph{Ledger PBFT}~(\LPBFT), a new BFT state machine replication protocol that stores ordered transactions in the ledger together with the protocol messages from replicas that justify the execution order. \LPBFT maintains Merkle trees~\cite{merkle-tree} over the ledger, and includes the roots of the trees in protocol messages. Since protocol messages are signed by the replicas, this commits them to the entire contents of the ledger.

\sys then issues \emph{receipts} to clients that provide succinct, universally-verifiable evidence that a transaction executed at a given position in the ledger. Receipts include signed protocol messages from multiple replicas that executed the transaction, thus binding them to a prefix of the ledger. 

Given a collection of receipts that violates linearizability, anyone can audit the ledger against the receipts to assign blame to at least $N/3$ replicas. Auditing produces an irrefutable \emph{universal proof-of-misbehavior}~(uPoM) in the form of contradictory statements signed by the same replica. The uPoM can be used by an \emph{enforcer}, \eg a court, to punish the members responsible for the misbehaving replicas. To provide accountability when all replicas misbehave, the enforcer may have to compel members to produce a ledger, imposing sanctions otherwise. While this formally adds a weak synchrony assumption, the enforcer chooses a conservative timeout to make blaming correct members unlikely in practice.

As an example of auditing, a client Alice may have a receipt for a transaction that executed at index~$i$ in the ledger and deposited \$1 million into client Bob's account. If Bob obtains the receipt from Alice and another receipt for a balance query transaction executed at index~$j$ ($j>i$) that does not show the balance, he may conduct an audit: he engages an enforcer to obtain the relevant ledger fragment, and replays the transactions between~$i$ and~$j$ to check for consistency with his receipts. If Bob is right, auditing produces a uPoM for at least $N/3$~replicas, which Bob sends to the enforcer to punish the consortium members responsible for the replicas.

To support changes to the consortium membership, \sys uses \emph{governance transactions} that alter the set of replicas and consortium members~\cite{russinovich2019ccf}. Governance transactions complicate receipt verification and auditing because they change the signing keys that must be considered. \sys therefore records governance transactions in the ledger, which allows clients, replicas, and auditors to determine the set of valid signing keys. Clients do not need to keep the full ledger, but only receipts of governance transactions. Since governance transactions are relatively rare, this \emph{governance sub-ledger} is significantly smaller than the full ledger.

Our \sys prototype provides individual accountability without compromising on throughput or latency: it implements a commitment scheme for transaction batches with only a single signature per replica. This enables clients to receive results with receipts after only two network round-trips. Our evaluation shows that \sys can execute over 47,000\unit{tx/s} with low latency.

The contributions of \sys and the paper structure are:

\begin{enumerate}[topsep=0pt]

\item \LPBFT, a BFT state machine replication protocol that orders and stores transactions together with the protocol messages justifying the execution order in a ledger~(\S\ref{subsec:lpbft-protocol}, \S\ref{subsec:lpbft-view-changes});

\item universally-verifiable client receipts that are generated efficiently with the ledger~(\S\ref{subsec:receipts});

\item an efficient auditing approach using the ledger and associated checkpoints, which produces short proofs-of-misbehavior~(\S\ref{sec:auditing}); and

\item a governance mechanism for changing members and replica sets, allowing auditing to assign blame even after members have left~(\S\ref{sec:governance}).

\end{enumerate}


\section{Overview of  \sys}

\Cref{fig:ledger-overview} shows \sys{}'s design. An \sys deployment provides a \emph{service}, with a well known name, to \emph{clients}, which are identified by their signing keys. Clients send requests to execute \emph{transactions} by calling stored procedures that define the service logic. Transactions are executed by \emph{replicas} against a strictly-serializable \emph{key-value store} that supports roll-back at transaction granularity. A transaction \emph{request}~$t$ reads and/or writes multiple key-value pairs and produces a transaction \emph{result}~$o$. 

Consortium \emph{members}, also identified by their signing keys, own the service. They may be added or removed over the service lifetime. For this, members issue \emph{governance transactions}, which change the consortium membership, add or remove replicas, and update stored procedures. The first governance transaction, the  \emph{genesis transaction}~$\mathit{gt}$, defines the initial members and replicas. Its hash is the service name.

\begin{figure}[tb]
  \centering
  \includegraphics[width=\columnwidth]{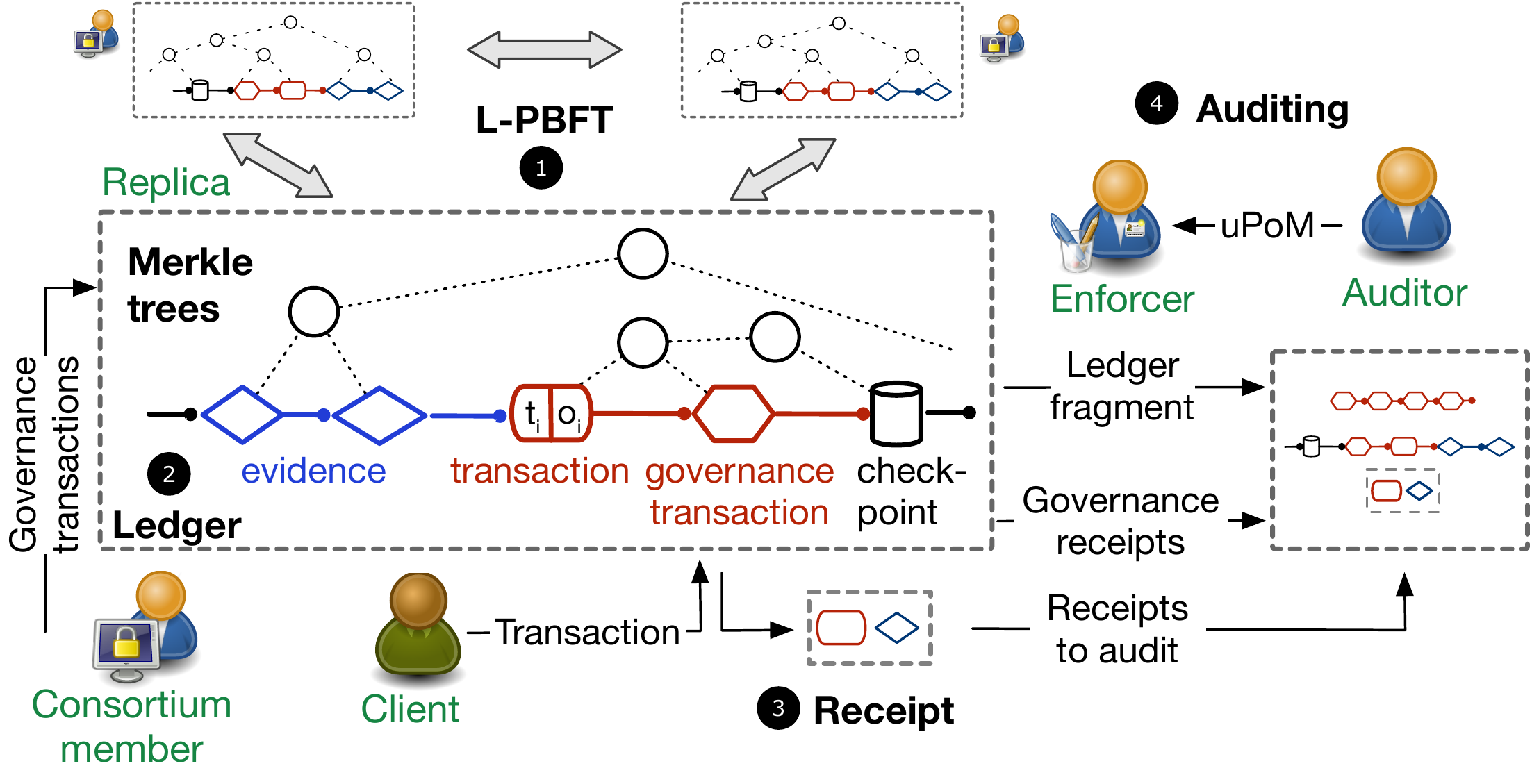}
  \caption{\sys permissioned ledger system}\label{fig:ledger-overview}
\end{figure}

\myparr{\myc{1}~Ledger PBFT (\LPBFT)} is a BFT state machine replication protocol used by replicas to order transactions. \LPBFT is based on PBFT~\cite{pbft}. It provides linearizability and liveness if at most $f = \lceil N / 3 \rceil - 1$ out of $N$~replicas fail in a partially-synchronous environment~\cite{partial-synchrony}.

\mypar{\myc{2}~Ledger} \LPBFT maintains an append-only \emph{ledger}, which stores each transaction request~$t$ and result~$o$ at a ledger index~$i$. Since the consortium membership and the replica set are dynamic, the ledger also records governance transactions. They form a \emph{governance sub-ledger}, which can be used to learn the public signing keys of active replicas and members at any index~$i$.

To assign blame, the ledger also includes \emph{evidence} that a transaction batch was committed by a quorum of replicas. This evidence consists of at least $N$$-$$f$~signed \LPBFT protocol messages for a batch. Finally, the ledger stores periodic checkpoints of the key-value store, allowing its state to be reconstructed by replaying the ledger from a checkpoint~$\mathit{cp}$.

All entries in the ledger are bound by Merkle trees. Protocol messages for a transaction batch contain the roots of the Merkle trees. This commits replicas to the whole ledger while allowing succinct existence proofs for entries.

\myparr{\myc{3}~Receipts} are created by replicas and returned to clients. They bind request execution to members via the replicas' signatures over Merkle tree roots that contains the executed request and the ledger's history. If two or more receipts are inconsistent with any linearizable execution, at least $f$$+$$1$~replicas must have signed contradictory statements and can thus be assigned blame.

More precisely, a receipt~$R$ for $\langle t, i, o \rangle$ states that request~$t$ was executed at index~$i$ and produced result~$o$. The receipt consists of $N$$-$$f$ protocol messages for $t$'s batch, signed by different replicas, and a path from a Merkle tree root to the leaf that contains an entry for $\langle t, i,o \rangle$.

Clients may obtain receipts from a reply to a request they sent, from replicas, or from other clients.
To validate a receipt, clients must check its signatures using the signing keys determined by the governance sub-ledger. A receipt therefore includes the ledger index of the last governance transaction, and clients must obtain the receipt of this governance transaction and all those preceding it. Clients cache governance transaction receipts and fetch missing ones from replicas.

\myparr{\myc{4}~Auditing} returns a \emph{universal proof-of-mis\-be\-havior}~(uPoM) if clients obtain receipts that are inconsistent with a linearizable execution. \sys{}'s ledger is universally-verifiable, \ie anyone can act as an \emph{auditor}: they replay the ledger, check consistency with receipts, and potentially generate a uPoM.

Since all consortium members and replicas may misbehave, an \emph{enforcer}, \eg a court, must compel members to produce a ledger copy for auditing, sanctioning non-compliance. The enforcer also punishes members based on uPoMs. It is unreasonable to assume that courts could run the service or audit long executions. Therefore, \sys only requires enforcers to re-execute transactions between two consecutive checkpoints to verify a uPoM in the worst case.

After a client passes a sequence of receipts and the governance sub-ledger to the auditor, the auditor confirms the receipts' validity by calculating a Merkle tree root and verifying the replica signatures. It then asks the enforcer to obtain the ledger fragment corresponding to the receipts from the replicas. The auditor checks the validity of the checkpoint~$\mathit{cp}$ referenced by the oldest receipt. It then replays the ledger from~$\mathit{cp}$, re-executing transaction requests while checking for consistency with receipts (including governance transaction receipts). If an inconsistency is found at index~$i$, the auditor creates a uPoM~$\langle i, \mathcal{F}, \mathit{cp}, R \rangle$ with a ledger fragment~$\mathcal{F}$, the checkpoint~$\mathit{cp}$, and the inconsistent receipt~$R$. The uPoM is then forwarded to the enforcer, which imposes penalties on the consortium members blamed.

\mypar{Threat model, and limitations} We assume a strong attacker that can compromise replicas, clients, auditors, and members to make them behave arbitrarily, but cannot break the cryptographic primitives. We trust the enforcer to assign blame to replicas and the members that operate them only when it verifies a valid uPoM or fails to obtain data for auditing. \sys provides linearizability and liveness if fewer than $1/3$ of the replicas are compromised~\cite{pbft}. With any number of compromised replicas, clients, auditors, and members, \sys never punishes members that operate only correct replicas unless they fail to provide data for auditing. 
In addition, \sys guarantees that at least $1/3$ of the replicas are blamed, and the members that operate them punished, if clients obtain receipts that are inconsistent with a linearizable execution. The current implementation does not 
prevent attacks that overwhelm the ledger with transactions to slow down auditing or replaying the governance sub-ledger.
It also does not blame replicas for liveness violations, e.g., not returning receipts. Possible defences include: having the enforcer timestamp the genesis transaction and bounding the rate of regular and governance transactions; and forwarding requests to the enforcer and having it monitor protocol execution to assign blame to replicas when receipts are not returned before a deadline. We leave the details of these defences for future work.


\section{\LPBFT protocol and receipts}
\label{sec:receipts}

Next, we describe how \LPBFT maintains a ledger with transactions and evidence~(\S\ref{subsec:lpbft-protocol}), and how it handles view changes~(\S\ref{subsec:lpbft-view-changes}). We then explain how evidence is used to create receipts~(\S\ref{sec:receipts:receipts}) and introduce performance optimizations~(\S\ref{sec:receipts:optimisations}). For ease of presentation, we first assume a fixed replica set; we add dynamic membership in \S\ref{sec:governance}.

\subsection{Protocol}
\label{subsec:lpbft-protocol}

\begin{algorithm}[tb]
  \footnotesize
  \caption{\LPBFTfull}
  \label{alg:lpbft}
  \SetAlgoNoEnd
  \DontPrintSemicolon 
  \SetNoFillComment
  \SetInd{0.4em}{0.4em}

  \KwOn({\FReceiveRequest{\scriptsize $t = \langle \mathsf{request}, a, c, H(\mathit{gt}), m_i \rangle_{\sigma_c}$}})
  {
    \label{alg:lpbft:receiveRequest}
    \scriptsize\textsf{
    \KwPre{\textsf{verify}($t$)}
      ${\mathcal T} \gets {\mathcal T} \; \cup \; \{t\}$ \;
  }}
  \KwOn({\FSendPrePrepare{}})
  {
    \label{alg:lpbft:sendPrePrepare}
    \scriptsize\textsf{
      \KwPre{\FIsPrimary{} $\wedge \; \mathsf{ready} \; \wedge \; |{\mathcal T}| > 0 \; \wedge $ \FHasEvidence{${\mathcal M}, v, s-P$}} 
      ${\mathcal B} \gets \left[ \right]; \; G \gets \{\}$ \;
      \ForEach{$t \in \mathcal T$}
      {
        ${\mathcal B} \gets {\mathcal B} \; || \; H(t) \;$; $\; \langle i,o \rangle \gets$ \FExecute{$kv, t$}; $\; G \gets G || \langle t,i,o \rangle$
      } 
      $\langle E_{s-P}, {\mathcal P}_{s-P}, {\mathcal K}_{s-P} \rangle \gets$ \FGetEvidence{${\mathcal M}, v, s-P$}\;
      ${\mathcal L} \gets {\mathcal L} \; || \; {\mathcal P}_{s-P} \; || \; {\mathcal K}_{s-P}; \; M \gets M || {\mathcal P}_{s-P} \; || \; {\mathcal K}_{s-P}$ \;
      ${\mathcal K}[v, s] \gets $\FCreateNonce{};$ \; \bar{M} \gets$ \FGetMerkleRoot{$M$}; $\bar{G} \gets$ \FGetMerkleRoot{$G$}\;
      $pp = \langle \mathsf{pre\mbox{-}prepare}, v, s, \bar{M}, \bar{G}, H({\mathcal K}[v,s]), E_{s-P} \rangle_{\sigma_{r}}$\;
      ${\mathcal L} \gets {\mathcal L} \; || \; pp \; || \; G; \; M \gets M \; || \; pp; \; {\mathcal M} \gets {\mathcal M} \; \cup \; \{pp\}$; $\; \; {\mathcal T} \gets \{\}$;   $\; \; s \gets s +1$\;
      \FSendToAllReplicas{$pp \; || \; {\mathcal B}$} \;
    }
  }
  \KwOn({\FReceivePrePrepare{\scriptsize $pp = \langle \mathsf{pre\mbox{-}prepare}, v, s', \bar{M}, \bar{G}, H(k), E_{s'-P} \rangle_{\sigma_r}, \mathcal B$}})
  {
    \label{alg:lpbft:receivePrePrepare}
    \scriptsize\textsf{
      \KwPre{\FIsBackup{} $\wedge$ \FVerify{$pp$} $\wedge \; \mathsf{ready} \; \wedge \; s'=s \; \wedge \; {\mathcal K}[v,s] = \mathsf{nil} \; \wedge$ \textsf{hasRequests}($\mathcal T, \mathcal B$)  $\wedge$ \FHasEvidence{${\mathcal M}, s'-P, E_{s'-P}$}}
      ${\mathcal M} \gets {\mathcal M} \; \cup \; \{pp\}; \; G \gets \{\}$ \;
      \ForEach{$h \in \mathcal B$}
      {
        $t \gets$ \textsf{removeTx}($h, \mathcal T$); $\; \langle i,o \rangle \gets$ \FExecute{$kv, t$};$ G \gets G || \langle t,i,o \rangle$ \; \label{alg:lpbft:receivePrePrepare:execute}
      }
      $\langle E_{s-P}, {\mathcal P}_{s-P}, {\mathcal K}_{s-P} \rangle \gets$ \FGetEvidence{${\mathcal M}, v, s-P, E_{s-P}$}\;
      ${\mathcal L} \gets {\mathcal L} \; || \; {\mathcal P}_{s-P} \; || \; {\mathcal K}_{s-P}; \; M \gets M || {\mathcal P}_{s-P} \; || \; {\mathcal K}_{s-P}; \;$ \; 
      \If{\FGetMerkleRoot{$M$} $\neq \bar{M}$ \KwOr \FGetMerkleRoot{$G$} $\neq \bar{G}$}
      {
        \label{alg:lpbft:receivePrePrepare:non-det}
        \textsf{undo}($pp, kv, \mathcal{M}, \mathcal{B}, \mathcal{T}, \mathcal{L}$); \KwRet{} \;
      }
      ${\mathcal L} \gets {\mathcal L} \; || \; pp \; || G; \; M \gets M \; || \; pp; \; {\mathcal K}[v, s] \gets $\FCreateNonce{}\; 
      $p = \langle \mathsf{prepare}, r, H({\mathcal K}[v, s]), H(pp) \rangle_{\sigma_{r}}$ \;
      \FSendToAllReplicas{$p$};  ${\mathcal M} \gets {\mathcal M} \; \cup \; \{p\} ; \; s \gets s+1$
  }}
  \KwOn({\FReceivePrepare{\scriptsize $p = \langle \mathsf{prepare}, r', H(k_{r'}), H(pp) \rangle_{\sigma_{r'}}$}})
  {
    \scriptsize\textsf{
      \KwPre{\FVerify{$p$}}
      ${\mathcal M} \gets {\mathcal M} \; \cup \; \{p\}$ \;
  }}
  \KwOn({\FBatchPrepared{\scriptsize $pp = \langle \mathsf{pre\mbox{-}prepare}, v, s', \bar{M}, \bar{G}, H(k_p), E_{s'-P} \rangle_{\sigma_p}$}})
  {
    \label{alg:lpbft:batchPrepared}
    \scriptsize\textsf{
      \KwPre{$ \mathsf{prepared}(pp, {\mathcal M}) \; \wedge \; \exists \langle \mathsf{prepare}, r', H({\mathcal K}[v,s']), H(pp) \rangle_{\sigma_{r'}} \in {\mathcal M}$}
      $c = \langle \mathsf{commit}, v, s', r, {\mathcal K}[v,s'] \rangle$ \;
      \FSendToAllReplicas{$c$}; $\; {\mathcal M} \gets {\mathcal M} \; \cup \; \{c\}$ \;
      \ForEach{$\langle t,i,o\rangle \in $\FGetTxForBatch{${\mathcal L}, v, s'$}}{
        \FSendReplyToClient{$t, \langle \mathsf{reply}, v, s', r, {\sigma_{r}}, {\mathcal K}[v,s'] \rangle$} \;
        \If{\FShouldSendReceipt{$r, t$}}
        {
          \label{alg:lpbft:batchPrepared:should-send-receipt}
          ${\mathcal S} \gets $\FGetMerklePath{$G, i$} \;
          \FSendReceiptToClient{$t, \langle \mathsf{replyx}, v, s', \bar{M}, H(k_p), E_{s'-P}, H(t), i, o, {\mathcal S} \rangle$} \;
        }
      }
    }
  }
  \KwOn({\FReceiveCommit{\scriptsize $c = \langle \mathsf{commit}, v, s', r', k_r\rangle$}})
  {
    \scriptsize\textsf{
      \KwPre{\FVerify{$c$}}
      ${\mathcal M} \gets {\mathcal M} \; \cup \; \{c\}$ \;
  }}
\end{algorithm} 

To support auditing, a BFT state machine replication protocol, such as PBFT~\cite{pbft}, must integrate with a ledger: it must ensure that replicas agree on a ledger with both transactions (requests and results) and protocol messages. It must also handle non-determinism to enable replaying the ledger. \LPBFT addresses this issue by agreeing on non-deterministic inputs~\cite{castro2002practical} and using \emph{early execution}: it requires the primary replica to propose a transaction result, which the backup replicas must agree on for the batch to commit. 
\LPBFT then maintains Merkle trees over all ledger entries and puts the trees' roots in protocol message, which ensures that all replicas agree on a serial history of the ledger.

\Cref{fig:ccf-pbft} gives an overview of \LPBFT with early execution: first clients send transaction requests to all replicas. The primary orders the requests, groups them into \emph{batches} and performs early execution. It then sends a \msg{pre-prepare} message to the backups, which includes the request batch and the execution results. Upon receiving the \msg{pre-prepare}, the backups execute the requests and confirm that the results match the primary's. If so, they send a \msg{prepare} message to all other replicas. After a replica receives a \msg{pre-prepare} and $N$$-$$f$$-$$1$ matching \msg{prepare} messages for the same sequence number~$s$ and view~$v$, the batch is \emph{prepared} at the replica at $v$ with $s$ if all batches with lower sequence numbers have also prepared. A replica then sends a \msg{reply} to the clients and \msg{commit} messages to the other replicas. We say that a batch is \emph{committed} at sequence number~$s$ if it has been prepared by $N$$-$$f$~replicas in the same view. A client has received a complete response when it has a \emph{receipt} consisting of replies from $N$$-$$f$ replicas.

A naive approach would require each replica to sign two protocol messages, \ie the \msg{pre-prepare}/\msg{prepare} and the \msg{commit} message, for each committed batch. Instead, \LPBFT uses a novel \emph{nonce commitment} scheme, in which replicas only sign the \msg{pre-prepare}/\msg{prepare} messages after including a hashed nonce. Instead of signing the \msg{commit}, a replica includes the unhashed nonce. This effectively halves the signatures that replicas emit to commit batches successfully. 

\Cref{alg:lpbft} presents the pseudocode of \LPBFT. The replica state includes: the current view~$v$ and batch sequence number~$s$; a set of transaction requests~${\mathcal{T}}$ waiting to be ordered; a message store~${\mathcal{M}}$; a nonce store~$\mathcal{K}$; a boolean~$\mathsf{ready}$ indicating if the replica can send/accept \msg{pre-prepare} messages; a replica identifier~$\mathsf{r}$; the key-value store~$kv$; the ledger~$\mathcal{L}$; and the Merkle tree~$M$ that binds the ledger entries.

In \textsf{receiveTransactionRequest} (line~\ref{alg:lpbft:receiveRequest}), a replica adds a \msg{request} message to $\mathcal{T}$, where $a$ identifies the invoked stored procedure and its arguments, $c$ is the client identifier, $H(\mathit{gt})$ is the genesis transaction hash, $m_i$ is the minimum index after which the request can be added to the ledger, and $\sigma_c$ is the client signature. $\sigma_c$ and $H(\mathit{gt})$ ensure that requests cannot be forged or moved to a different ledger, and $m_i$ allows clients to create an ordering dependency between the request and a previously executed transaction.

The function~\textsf{sendPrePrepare} (line~\ref{alg:lpbft:sendPrePrepare}) uses early execution to include the execution result in the batch's Merkle tree root. The primary~$p = v\ \mathrm{mod}\ N$ collects a batch of transaction requests, executes them, and appends them to a new Merkle tree~$G$. Then, the primary retrieves the commitment evidence ${\mathcal{P}}_{s-P}$ and ${\mathcal{K}}_{s- P}$ for the batch at $s$$-$$P$ from the message store~$\mathcal{M}$ and appends it to the ledger. $E_{s-P}$ is a bitmap that records the replicas that supplied commitment evidence.

\begin{figure}[tb]
  \centering
  \includegraphics[width=.9\columnwidth]{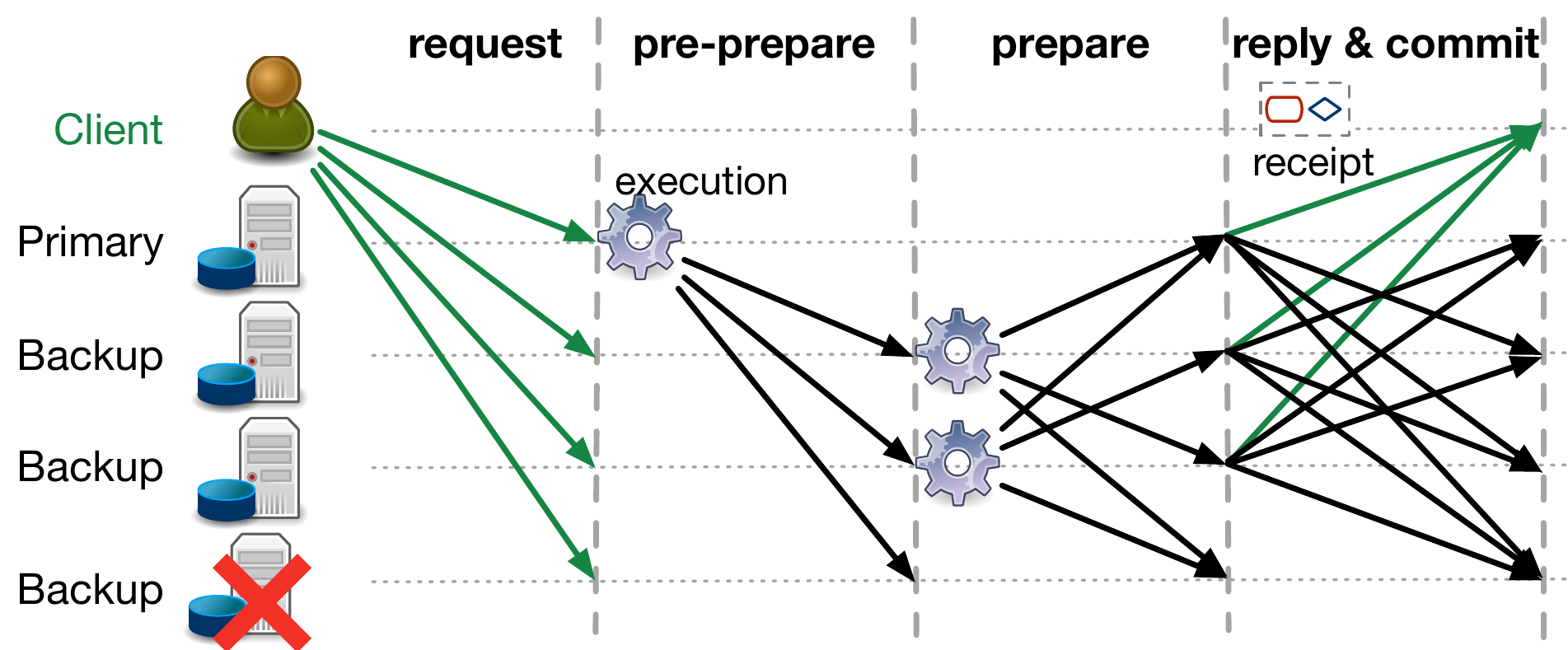}
  \caption{\LPBFT protocol with early execution and receipts}
  \label{fig:ccf-pbft}
\end{figure}

Next, the primary creates the \msg{pre-prepare} message with the hash of a fresh nonce ${\mathcal{K}}[v,s]$, the root of the Merkle trees, $\bar{M}$ and $\bar{G}$, and signs it. $G$ is a Merkle tree that contains all $\langle t, i, o \rangle$ entries in a batch. The complete \msg{pre-prepare} message has two extra fields: $i_g$, the index of the last governance transaction, which allows clients to verify receipts with a changing set of replicas~(see~\S\ref{sec:governance:receipts_auditing}); and $d_C$, a digest of the key-value store state at the last checkpoint, which enables auditing from a checkpoint without replaying the ledger from the start~(see~\S\ref{sec:auditing}).

By signing $\bar{M}$, the primary commits to the contents of the ledger, including the commitment evidence for $s$$-$$P$ that it retrieved and added to the ledger. It is important for the primary to order the evidence to ensure that replicas agree on the ledger: if replicas added their own evidence to the ledger when they received \msg{prepare} and \msg{commit} messages, their ledgers could diverge. The commitment evidence ${\mathcal{P}}_{s-P}$ contains $N$$-$$f$$-$$1$ \msg{prepare} messages for sequence number $s$$-$$P$ and view~$v$ that match the \msg{pre-prepare} at sequence number $s$$-$$P$ in the ledger. ${\mathcal{K}}_{s-P}$ are the $N$$-$$f$ nonces with hashes in the \msg{pre-prepare}/\msg{prepare} messages in ${\mathcal{P}}_{s-P}$. This evidence is sufficient to prove to a third party that the batch at $s$$-$$P$ prepared at $N$$-$$f$ replicas and therefore committed with $s$. The \msg{pre-prepare} message along with the leaves of $G$ are then added to the ledger.

The primary communicates its ordering decision by sending the \msg{pre-prepare} message to all replicas, together with a list~$\mathcal{B}$ of the hashes of transaction requests in execution order. The requests are sent separately by the clients, and the commitment evidence for $s$$-$$P$ is not included in the message. The \msg{pre-prepare} messages are $\mathrm{O}(N)$ in size but the constant is small. Our implementation uses 8\unit{bytes} in the $E_{s-P}$ bitmap to support up to 64~replicas, making the \msg{pre-prepare} messages effectively $\mathrm{O}(1)$. 

\Cref{fig:chained-merkle-trees} gives an example of the ledger state after this step. For each transaction in the batch, the primary adds a ledger entry in the order executed. The entry for $T_i$ has the form $\langle t, i, o \rangle$ where $o$ includes the reply sent to the client and the hash of the transaction's write-set; $pp_{s}$ is the \msg{pre-prepare} for $s$, and ${\mathcal{P}}_{s-P}$ and ${\mathcal{K}}_{s- P}$ are evidence that the batch at sequence number $s$$-$$P$ committed. \LPBFT pipelines the ordering of up to $P \ge 1$ concurrent batches to improve performance. Therefore, the commitment evidence lags $P$ behind $s$, because it is unavailable when the primary sends the \msg{pre-prepare} for~$s$. \Cref{proof:lpbft}, \Cref{lemma:early-execution} shows that \emph{early execution} maintains linearizability.

\begin{figure}[tb]
\centering
\includegraphics[width=.9\columnwidth]{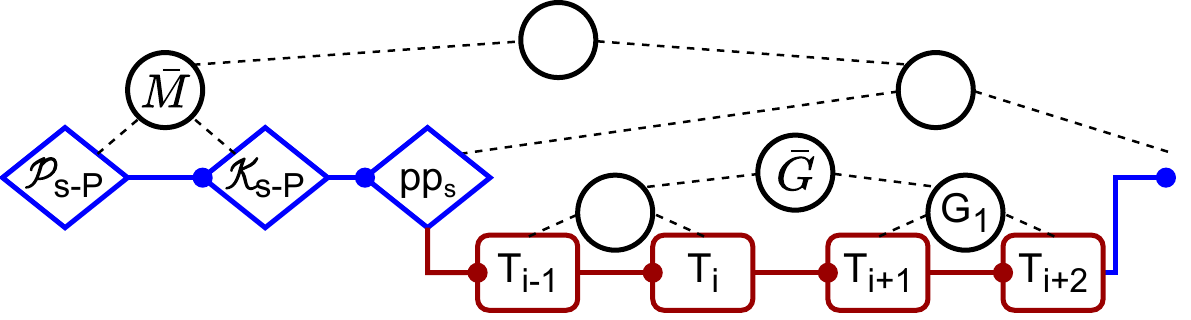}
\caption{Ledger with evidence and Merkle trees}\label{fig:chained-merkle-trees}
\end{figure}

When a backup replica receives the \msg{pre-prepare} (line~\ref{alg:lpbft:receivePrePrepare}), it rejects the message if it already sent a \msg{prepare} for the same view and sequence number (${\mathcal{K}}[v,s] \neq \mathsf{nil}$). Otherwise, it checks if it already has the requests and commitment evidence referenced by the \msg{pre-prepare}. Replicas store received requests, \msg{prepare}, and \msg{commit} messages in non-volatile storage ($\mathcal{M}$) until they receive (or send) a corresponding \msg{pre-prepare}. To reduce network load, the primary does not resend requests or messages used as commitment evidence. If the backup is missing messages, it requests that the primary retransmit them, because a correct primary is guaranteed to have them.

The backup then executes the requests in the order prescribed by the primary, and adds the resulting transaction entries to a new Merkle tree $G$~(line~\ref{alg:lpbft:receivePrePrepare:execute}). Then, it adds the same ${\mathcal{P}}_{s-P}$ and ${\mathcal{K}}_{s-P}$ as the primary to the ledger. At this point, the ledger at the backup should be identical to the one at the primary just before the \msg{pre-prepare} message is added. The backup checks that the roots of its Merkle trees match $\bar{M}$ and $\bar{G}$ in the \msg{pre-prepare}, respectively.  If not, the message is rejected, the entries for batch~$s$ are removed from the ledger, and the transactions are rolled back. Otherwise, the backup adds the \msg{pre-prepare} to the ledger, followed by the leaves of the Merkle tree $G$, and sends a matching \msg{prepare} message with the format $\langle \mathsf{prepare}, \mathsf{r}, H({\mathcal{K}}[v,s]), H(\mathit{pp}) \rangle_{\sigma_{\mathsf{r}}}$, where $H({\mathcal{K}}[v,s])$ commits a fresh nonce, and $H(\mathit{pp})$ is the \msg{pre-prepare}{}'s hash.

\LPBFT ensures deterministic transaction execution by agreeing on non-deterministic inputs~\cite{castro2002practical}. Line~\ref{alg:lpbft:receivePrePrepare:non-det} ensures that a backup's execution of batch~$\mathcal{B}$ and its ledger are identical to those of the primary by comparing the Merkle roots~$\bar{G}$ and $\bar{M}$. If this check fails, the backup rolls back execution and attempts to view change (\S\ref{subsec:lpbft-view-changes}). This way divergent execution due to bugs, i.e., failing to identify non-deterministic inputs, can affect liveness but not diverge the ledger.

In \textsf{batchPrepared} (line~\ref{alg:lpbft:batchPrepared}), the nonce commitment and early execution allow replicas to return replies to clients in two message round trips without signing \msg{reply} or \msg{commit} messages. When the batch prepares at replica~$r$, it sends a \msg{commit} message with the format $\langle \mathsf{commit}, v, s', r, {\mathcal{K}}[v,s']\rangle$ where ${\mathcal{K}}[v,s']$ is the nonce the replica committed to in the \msg{pre-prepare}/\msg{prepare} messages that it sent for $v$ and $s'$. Since the nonce ${\mathcal{K}}[v,s']$ is revealed to clients and replicas only when a replica prepares the batch having a \msg{pre-prepare}/\msg{prepare} message and the corresponding nonce can prove to a third party that the replica prepared the batch at $v$ and $s'$ (see \Cref{proof:lpbft}, \Cref{lemma:nonce-commitment}).

Finally, a replica~$r$ commits a prepared batch $v, s'$ after it receives $N$$-$$f$ \msg{commit} messages, including its own. The nonce hashes in the \msg{commit} messages must match the ones in the \msg{pre-prepare}/\msg{prepare} messages.

We prove that \LPBFT produces a linearizable execution order in \Cref{proof:lpbft}, \Cref{lemma:linearizability}.

\subsection{View changes}
\label{subsec:lpbft-view-changes}

\begin{algorithm}[tb]
  \footnotesize
  \caption{View Changes in \LPBFT}
  \label{alg:lpbft-view-change}
  \SetAlgoNoEnd
  \DontPrintSemicolon 
  \SetNoFillComment
  \SetInd{0.4em}{0.4em}
  \KwOn({\FSendViewChange{}})
  {
    \label{alg:lpbft-view-change:sendViewChange}
    \scriptsize\textsf{
      \KwPre{\textsf{primaryAppearsFaulty}($v$)}
      ${\mathcal{PP}} = \mathsf{getPLastPrepared}(\mathsf{msgs}({\mathcal L}) \; \cup \; {\mathcal M})$ \;
      \label{alg:lpbft-view-change:sendViewChange:getPLastPrepared}
      $v = v+1; \; \mathsf{ready} \gets \mathsf{false}; vc = \langle \mathsf{view\mbox{-}change}, v, r, {\mathcal{PP}} \rangle_{\sigma_{r}}$ \;
      \label{alg:lpbft-view-change:ready-false:a}
      \FSendToAllReplicas{$vc$}; $\; \; {\mathcal M} \gets {\mathcal M} \; \cup \; \{vc\}$ \; 
  }}
  \KwOn({\FReceiveViewChange{\scriptsize $vc = \langle \mathsf{view\mbox{-}change}, v', r', {\mathcal{PP}} \rangle_{\sigma_r}$}})
  {
    \label{alg:lpbft-view-change:receiveViewChange}
    \scriptsize\textsf{
      \KwPre{ $v' >= v \; \wedge \; \mathsf{verify}(vc) \; \wedge $ \FHasPrepares{$\mathsf{msgs}({\mathcal L}) \; \cup \; {\mathcal M}, \mathsf{getLast}({\mathcal{PP}})$}}
      ${\mathcal M} \gets {\mathcal M} \; \cup \; \{vc\}$ \;
      \If {$|\mathsf{getViewChanges}({\mathcal M}, v')| > f \; \wedge \; v' > v$}
      {
      $v = v'-1;\ \mathsf{setPrimaryApearsFaulty()} $ \;
      \FSendViewChange{} \;
      \label{alg:lpbft-view-change:ready-false:b}
      }
   }}
  \KwOn({\FSendNewView{\scriptsize $v$}})
  {
    \label{alg:lpbft-view-change:sendNewView}
    \scriptsize\textsf{
      \KwPre{\FIsPrimary{$v$} $\wedge \; \neg\mathsf{ready}\; \wedge \; |\mathsf{getViewChanges}({\mathcal M}, v)| > N-f$}
      $\langle \bar{M}, E_{vc}, h_{vc}, {\mathcal{PP}}_{ov} \rangle = \mathsf{processViewChanges}(\mathsf{getViewChanges}({\mathcal M}, v))$ \;
      \label{alg:lpbft-view-change:sendNewView:processViewChanges}
      $nv = \langle \mathsf{new\mbox{-}view}, v, \bar{M}, E_{vc}, h_{vc} \rangle_{\sigma_{r}}; \; {\mathcal L} \gets {\mathcal L} \; || \; nv ; \; M \gets M \; || \; nv$\;
      \FSendToAllReplicas{$nv$}  \;
      $\mathsf{resendPreparesInNewView}({\mathcal{PP}}_{ov})$; $\mathsf{ready} \gets \mathsf{true}$ \;
  }}
  \KwOn({\FReceiveNewView{\scriptsize $nv = \langle \mathsf{new\mbox{-}view}, v, \bar{M}, E_{vc}, h_{vc} \rangle_{\sigma_{r'}}, {\mathcal{PP}}_{nv}$}})
  {
    \label{alg:lpbft-view-change:receiveNewView}
    \scriptsize\textsf{
      \KwPre{\FIsPrimary{$r', v$} $\wedge$ \textsf{hasRequests}(${\mathcal T}, {\mathcal{PP}}_{nv}$)  $\wedge$ \FHasEvidence{${\mathcal M}, {\mathcal{PP}}_{nv}$}  $\wedge \; r' \neq r \; \wedge \; \neg\mathsf{ready}\; \wedge \; |\mathsf{getViewChanges}({\mathcal M}, E_{vc}, h_{vc})| > N-f$}
      $\langle \bar{M'}, {\mathcal{PP}}'_{ov} \rangle = \mathsf{processViewChanges}(\mathsf{getViewChanges}({\mathcal M}, E_{vc}, h_{vc}))$ \;
      \If{$\bar{M'} = \bar{M}$}
      {
         ${\mathcal L} \gets {\mathcal L} \; || \; nv; \; M \gets M \; || \; nv$\;
         \lIf{$\mathsf{ready} \gets \mathsf{processPreparesInNewView}({\mathcal{PP}}_{nv},  {\mathcal{PP}}'_{ov})$} {\KwRet{}}
      }
     $\mathsf{undo}(nv, s, {\mathcal M}, {\mathcal L})$ \;
  }}
\end{algorithm} 

During the \LPBFT protocol execution, the primary may misbehave or be slow, which requires a \emph{view change}. The change of the primary must be done in a manner that does not preclude auditing, which is a new requirement that goes beyond PBFT's view change protocol. \LPBFT view changes are auditable and must provide proof that a batch's re-execution produces the same result as the original execution.

\LPBFT addresses this as follows: it sends the evidence that batches prepared during view changes and includes the Merkle tree root $\bar{G}$ of a batch and its execution in the \emph{pre-prepare} message, which ensures that batches are re-executed consistently. During a view change, each replica sends a \msg{view-change} message with information about prepared requests. The primary for a new view~$v'$ sends a \msg{new-view} message backed by $N$$-$$f$~\msg{view-change} messages for $v'$. For each sequence number with a prepared batch in the \msg{view-change} messages, the primary picks the batch that prepared with the largest view and proposes it in $v'$. Since all committed requests have also prepared, this ensures linearizability with batch execution ordered by the sequence numbers at which batches committed.

\Cref{alg:lpbft-view-change} formalizes the pseudocode for view changes. If the primary for view~$v$ appears faulty or slow, a replica sends a \msg{view-change} message, $\langle \mathsf{view\mbox{-}change}, v+1, \mathsf{r}, {\mathcal{PP}} \rangle_{\sigma_{\mathsf{r}}}$, to all other replicas (line~\ref{alg:lpbft-view-change:sendViewChange}), where ${\mathcal{PP}}$ contains the last $P$ \msg{pre-prepare} messages that prepared locally (line~\ref{alg:lpbft-view-change:sendViewChange:getPLastPrepared}). Only the last message in ${\mathcal{PP}}$ is required to provide linearizability, because it includes the Merkle tree roots~$\bar{M}$ and $\bar{G}$ that determine the ledger contents up to that point. The other \msg{pre-prepare} messages are used during auditing to verify that replicas reported the batches they prepared in \msg{view-change}~(\S\ref{sec:auditing}).

When replicas receive a \msg{view-change} message (line~\ref{alg:lpbft-view-change:receiveViewChange}), before processing it, they fetch missing \msg{prepare} messages from the sender to prove that the last \msg{pre-prepare} in ${\mathcal{PP}}$ has prepared. When replicas increment $v$, they set $\mathsf{ready}$ to $\mathsf{false}$ (lines~\ref{alg:lpbft-view-change:ready-false:a}, \ref{alg:lpbft-view-change:ready-false:b}), which ensures that they do not send or accept \msg{pre-prepare} messages until they have completed the \msg{new-view}.

After accepting $N$$-$$f$ \msg{view-change} messages for the new view (line~\ref{alg:lpbft-view-change:sendNewView}), the new primary calls \textsf{processViewChanges}, which picks the \msg{view-change} message~$\mathit{vc}_{\mathit{lp}}$ with the last prepared \msg{pre-prepare} message~$\mathit{pp}_{\mathit{lp}}$ from those with the largest view number. It then updates the ledger to match the Merkle roots in $\mathit{pp}_{\mathit{lp}}$ by fetching missing ledger entries from replicas that sent matching \msg{prepare} messages. Since at least $f$$+$$1$ of those are correct, this is always possible. The primary checks that all messages in ${\mathcal{PP}}$ of $\mathit{vc}_{\mathit{lp}}$ appear at the right ledger positions; if not, it discards $\mathit{vc}_{\mathit{lp}}$ and re-tries (omitted from \Cref{alg:lpbft-view-change}).
 
Next the primary resets the ledger to $s_{\mathit{lp}}$$-$$P$, because the batches up to this point are guaranteed to have committed. It saves all the request batches and commitment evidence for sequence numbers between $s_{\mathit{lp}}$$-$$P$ and $s_{\mathit{lp}}$ and returns it in ${\mathcal{PP}}_{\mathit{ov}}$. This is needed to resend \msg{pre-prepare} messages for the prepared batches in the new view. The function ends by adding an entry with the $N$$-$$f$ \msg{view-change} messages that it accepted to the ledger in order of increasing replica identifier; $h_{\mathit{vc}}$ is the hash of that entry and $E_{\mathit{vc}}$ is a bitmap with the replicas that sent the messages. It returns the root of the Merkle tree~$\bar{M}$, $E_{\mathit{vc}}$, $h_{\mathit{vc}}$, and ${\mathcal{PP}}_{\mathit{ov}}$ (line~\ref{alg:lpbft-view-change:sendNewView:processViewChanges}). The primary appends the \msg{new-view} to the ledger, sends it to all replicas, resends the prepared batches in \msg{pre-prepare} messages in the new view, and adds them to the ledger.

When backups receive the \msg{new-view} (line~\ref{alg:lpbft-view-change:receiveNewView}), they obtain missing \msg{view-change} messages, requests and evidence that it references, and call \textsf{processViewChanges}. If it returns a Merkle tree root equal to the one in \msg{new-view}, they accept the message, add it to the ledger, and process the \msg{pre-prepare} messages ${\mathcal{PP}}_{\mathit{nv}}$. If these match the batches and evidence in ${\mathcal{PP}}'_{\mathit{ov}}$ for the same sequence numbers, they are added to the ledger; otherwise, all changes are undone.

\subsection{Receipts}
\label{sec:receipts:receipts}
\label{subsec:receipts}

\begin{algorithm}[tb]
  \footnotesize
  \caption{Verifying Receipts}
  \label{alg:receipt}
  \SetAlgoNoEnd
  \DontPrintSemicolon 
  \SetNoFillComment
  \SetInd{0.4em}{0.4em}
  \vspace*{-0.25em}
    \KwOn({\FVerifyReceipt{\scriptsize $\langle t,i,o \rangle, \langle v, s, \bar{M}, H(k_p), E_{s-P}, i_{g}, d_C), \sigma_p, E_s, \Sigma_s, {\mathcal{K}}_s, {\mathcal{S}} \rangle $}})
    {
      \scriptsize\textsf{
        $\bar{G'} \gets$ \FPathHash($\langle t,i,o \rangle$)  \;
            \label{alg:receipt:merkle:start}
        \ForEach{$G_{i} \in {\mathcal{S}}$}
        { 
          $\bar{G'} \gets$ \FPathHash($\bar{G'}$, $G_{i}$) \;
            \label{alg:receipt:merkle:end}
        } 
        $pp=\langle \mathsf{pre\mbox{-}prepare}, v, s, \bar{M}, \bar{G'}, H(k_p), E_{s-P}, i_{g}, d_C \rangle$\;
        \lIf {\KwNot \FCheckSignature($\sigma_p$, pp)} 
          { 
            \label{alg:receipt:checkpp}
            \KwRet{false}
          }
        \ForEach{$r \in E_s$}
        {
          \lIf {$r=p \wedge H({\mathcal{K}}_s[p]) \neq H(k_p)$}
          {
             \KwRet{false}
          }
          \lIf {$r \neq p \wedge$ \KwNot \FCheckSignature($\Sigma_s[r]$, $\langle \mathsf{prepare}, r, H({\mathcal{K}}_s[r]), H(pp_{\sigma_p}) \rangle$)} 
          {
            \label{alg:receipt:checkp}
            \KwRet{false}
          }
        }  
        \KwRet{true} \;
    }}
\end{algorithm} 

To allow third parties to audit the ledger against the transaction results returned to clients, \LPBFT returns \emph{receipts}, which are statements signed by $N$$-$$f$ replicas that a transaction request~$t$ executed at index~$i$ and produced a result~$o$.  \LPBFT exploits the per batch Merkle tree~$G$ together with the nonce commitment scheme (\S\ref{subsec:lpbft-protocol}) to avoid having replicas sign the reply for each request.

\label{subsubsec:receipts:creating}

\mypar{Creating receipts} When a transaction batch described by \msg{pre-prepare} $\mathit{pp}$ prepares at replica~$r$, view~$v$ and sequence number~$s'$ (\Cref{alg:lpbft}, line~\ref{alg:lpbft:batchPrepared}), it sends $\langle \mathsf{reply}, v, s', r, {\sigma_{\mathsf{r}}}, {\mathcal{K}}[v,s'] \rangle$ to every client with a transaction in the batch. (If the client has multiple transactions in the batch, only one reply is sent.) By revealing the nonces, the replicas provide the client with proof that they claimed to have prepared the batch without a signed reply.

Only a designated replica, chosen based on $t$, sends the result and the rest of the receipt to the client~(line~\ref{alg:lpbft:batchPrepared:should-send-receipt}). The replica computes a list of sibling hashes $\mathcal{S}$ along the path from the leaf to the root of the per-batch Merkle tree~$G$. For the example of $T_i$ in \Cref{fig:chained-merkle-trees}, $\mathcal{S}$ consists of the digest of $T_{i-1}$ and $G_1$, which is sufficient to recompute $\bar{G}$ given $T_i$. It then sends the client $\langle \mathsf{replyx}, v, s', \bar{M}, H(k_p), E_{s'-P}, i_{g}, d_C, H(t), i, o, {\mathcal{S}} \rangle$, where $i_{g}$ and $d_C$ are used for auditing.

\label{subsubsec:receipts:verifying}

\mypar{Verifying receipts} The client waits for $N$$-$$f$~replicas to send \msg{reply} messages with the same $v$ and $s$, and for a \msg{replyx} message with the same $v$ and $s$. It then recreates the \msg{pre-prepare} and \msg{prepare} messages (\Cref{alg:receipt}, line~\ref{alg:receipt:checkpp}), with the information in \msg{replyx} and the hashes of the nonces, and verifies the signatures. (We describe how to determine $N$ and verify signatures under dynamic membership in \S\ref{sec:governance:receipts_auditing}.) This step is shared across all transaction requests that the client may have sent in the batch.

\sys uses the Merkle tree~$G$ to bind signatures in \msg{pre-prepare} and \msg{prepare} messages to transactions in the batch, enabling replicas to produce a single signature per batch. In the example in \Cref{fig:chained-merkle-trees}, the client checks if $\bar{G} = H(H( H(T_{i-1}) || H(\langle t,i,o \rangle)) || G_1))$ (lines~\ref{alg:receipt:merkle:start}--\ref{alg:receipt:merkle:end}). If the hashes match, the client has a valid receipt, \ie a statement signed by $N$$-$$f$ replicas that a request~$t$ executed at index~$i$ and produced a result~$o$; otherwise (or if the client does not receive replies before a timeout), it retransmits the request and selects a different replica to send back \msg{replyx}. (The application is responsible for ensuring exactly-once semantics if needed.)

Clients store the receipt for $\langle t, i, o \rangle$ as $\langle v, s, \bar{M}, H(k_p), E_{s-P},$ $i_{g}, d_C, \sigma_p, E_s, \Sigma_s, {\mathcal{K}}_s, {\mathcal{S}} \rangle$ where $\Sigma_s$ is a list of the signatures in \msg{prepare} messages, ${\mathcal{K}}_s$ is a list of nonces, and $E_s$ is a bitmap indicating the replicas with entries in $\Sigma_s$, and ${\mathcal{K}}_s$, sorted in increasing order of replica identifier. All receipt components, including common hashes in ${\mathcal{S}}$, are shared across requests in the same batch. 

Clients must store the receipts together with the transaction request and the corresponding result to resolve future disputes. This is not a burden because receipts are concise: all components have constant size, except $|{\mathcal{S}}|$, whose number of entries is logarithmic in the number of requests in a batch; $\Sigma_s$ and ${\mathcal{K}}_s$ have up to $N$$-$$f$ entries. In addition, most intermediate hashes in ${\mathcal{S}}$ can be shared across collections of receipts. We explored using signature aggregation~\cite{aggrsigs} to reduce the size of $\Sigma_s$, but, for realistic consortia sizes, verifying the signatures becomes more expensive than our current implementation.

\label{sec:receipts:Ledger-structure}

\subsection{Performance optimizations}
\label{sec:receipts:optimisations}

\LPBFT includes several optimizations to improve transaction and auditing throughput.

\myparr{Checkpoints} in \LPBFT allow new replicas to start processing requests without having to replay the ledger from the start (\S\ref{sec:governance:reconfiguration}); slow replicas to be brought up-to-date using a recent checkpoint; and auditing to start from a checkpoint instead of the beginning of the ledger(\S\ref{sec:auditing:process}).

Checkpoints include the key-value store and the Merkle tree $M$'s newest leaf, root, and the connecting branches. Replicas create a checkpoint~$\mathit{cp}_s$ when they execute a batch with sequence number~$s$ such that $s\ \mathrm{mod}\ C = 0$. The primary adds a batch to the ledger at sequence number $s$$+$$C$ with a special \emph{checkpoint transaction}, which records the checkpoint digest. $C$ is chosen to give replicas enough time to complete a checkpoint without delaying \LPBFT execution. Backups only accept the \msg{pre-prepare} for $s$$+$$C$ if they compute the same checkpoint digest for sequence number~$s$.

When a replica fetches checkpoint~$\mathit{cp}_s$, it also retrieves the ledger up to $s$. It does not need to replay the ledger or check all signatures (with the exception of governance transactions; \S\ref{sec:governance:receipts_auditing}). Instead, it checks the signatures in checkpoint receipts and that the ledger contents between consecutive checkpoints are consistent with the Merkle tree roots in the corresponding receipts. This is done from the start of the ledger until $s$$+$$C$.

\mypar{Cryptography} \LPBFT reduces the impact of cryptographic operations. Signature verification is parallelized for messages received from replicas and clients~\cite{aardvark, bessani2014state} to improve throughput and scalability. All messages are sent over encrypted and authenticated connections, even signed messages. This mitigates denial-of-service attacks that consume replica resources verifying signatures~\cite{aardvark}. 

To further improve performance, backups overlap the execution of request batches with the validation of \msg{pre-prepare} signatures. They only send the \msg{prepare} after both completed. Since \msg{pre-prepare} messages are received over authenticated connections, this always succeeds for correct primaries.


\section{Auditing and enforcement}
\label{sec:auditing}

In this section, we describe how auditing produces \emph{universal proofs-of-misbehavior}~(uPoMs) when linearizability is violated~(\S\ref{subsec:auditing}), and the role of the enforcer in obtaining ledgers for auditing and punishing the members responsible for misbehaving replicas~(\S\ref{subsec:upom}). We first focus on the simpler case of auditing without governance transactions; \S\ref{sec:governance} describes governance transactions and their impact on auditing.

\subsection{Auditing}
\label{sec:auditing:process}
\label{subsec:auditing}

An audit is triggered when someone, usually a client, obtains a sequence of transaction receipts that violate linearizability, \ie when no linearizable execution of the stored procedures that define the transactions can produce the sequence of receipts. The mechanism to detect linearizability violations is application dependent. It involves clients, which interact through a sequence of transactions, exchanging receipts and using the application semantics to reason about the correctness of the receipt sequence. We describe a banking-inspired example in the introduction.

The goal of auditing is to detect dishonest behavior regardless of the number of misbehaving replicas, \ie it must find proof of misbehavior even if all replicas collude and rewrite the ledger. \sys therefore tightly integrates the ledger with receipts---even if the ledger is rewritten, the misbehaving replicas are unable to alter the receipts.

An audit can be performed by anyone, and begins when an \emph{auditor} receives a collection of receipts. Next, the auditor requests a \emph{checkpoint} and a \emph{ledger fragment} that contains the section of the ledger spanning the receipts. Any honest replica that signed the receipts is guaranteed to have the checkpoint and ledger fragment. When the auditor receives the requested data, it verifies the ledger structure by checking the protocol messages and their order, and validating any signatures in the ledger---but it does not re-execute transactions. Then, the auditor checks that the transactions referenced by the receipts are present at the right positions in the ledger.

If the above steps have not discovered misbehavior, there remains the possibility that at least $N-f$ of the replicas colluded and agreed on an incorrect execution result. Therefore, the auditor loads the checkpoint and replays the transactions from the ledger fragment to check if execution results are correct. Throughout this process, if dishonest behavior is uncovered, the auditor can produce a universally-verifiable proof that at least $f$$+$$1$ replicas misbehaved.

More formally, \Cref{alg:auditing-with-receipts} presents the pseudocode for the auditing process. First, the auditor receives an ordered set of receipts~$\mathcal{R}=\{ \langle \langle t_0, i_0, o_0 \rangle, x_0 \rangle,\ldots,\langle \langle t_k, i_k, o_k \rangle, x_k \rangle \}$ where $k \ge 1$ and $\forall l \in [0, k) : s_l \le s_{l+1}$. Here, $s_i$ is the sequence number that is specified in $x_i$. The auditor invokes \textsf{auditReceipts} (line~\ref{alg:auditing-audit-receipts}) to check if the receipts are valid and the minimum index requirements have been satisfied. If there is a receipt that violates the requirement in the request, all replicas that have signed the receipt can be blamed.
 
After that, the auditor must obtain a ledger fragment and checkpoint that are \emph{complete} in relation to $\mathcal{R}$ (line~\ref{alg:auditing-get-ledger}). We formally define completeness in \Cref{sec:auditing-correctness}, but intuitively the ledger fragment must be (i)~\emph{well-formed}; (ii)~include all batches and evidence between sequence numbers~$s_{C_0}$ and $s_k$ where $s_{C_0}$ is the sequence number of the checkpoint transaction that is linked in the first receipt; and (iii)~include \msg{view-change} messages for all views in $\mathcal{R}$. The transaction and checkpoint at $s_{C_0}$ must match the checkpoint linked in the first receipt.
A ledger fragment is \emph{valid} if it can be produced by a sequence of correct primaries in a sequence of views where there are at most $f$~Byzantine failures. It is well-formed if it is valid, or if it would be valid if not for the incorrect execution of some transactions and/or checkpoints.
A correct replica always maintains a well-formed ledger.
 
In \textsf{getCheckpointAndLedger} (line~\ref{alg:auditing-get-ledger}), the auditor, with the help of an \emph{enforcer}, obtains ledger fragments and checkpoints from replicas that signed the latest receipt with the highest view number in $\mathcal{R}$ (line~\ref{alg:auditing-enforcer-get-ledger}). The auditor checks if  responses are complete in relation to the receipts. If a ledger fragment is not well-formed or misses the required \msg{view-change} messages, the auditor can blame the responding replica. Below, we assume that the responses contain no invalid signatures, we show in \Cref{sec:auditing-correctness} how the auditor handles that case.

\begin{algorithm}[tb]
  \footnotesize
  \caption{Ledger Auditing (simplified)}
  \label{alg:auditing-with-receipts}
  \SetAlgoNoEnd
  \DontPrintSemicolon 
  \SetNoFillComment
  \SetInd{0.4em}{0.4em}
    \KwOn({\FAudit{$\mathcal{R} = \{ \langle \langle t_0, i_0, o_0 \rangle, x_0 \rangle, \ldots, \langle \langle t_k, i_k, o_k \rangle, x_k \rangle \}$}}) 
    {
      \scriptsize\textsf{
        \FAuditReceipts{$\mathcal{R}$}\;
        \label{alg:auditing-audit-receipts}
        $C_0, s_{C_0}, \mathcal{L} \gets $\FGetCheckpointAndLedger{$x_0, x_k$}\;
        \label{alg:auditing-get-ledger}
        \FVerifyReceiptsInLedger{$\mathcal{R}, \mathcal{L}$}\;
        \label{alg:auditing-verify-receipts}
        \FReplayLedger{$C_0, s_{C_0}, \mathcal{L}$}\;
        \label{alg:auditing-replay-no-error}
        \label{alg:auditing-replay-end}
    } }
    \KwOn({\FAuditReceipts{$\mathcal{R} = \{ \langle \langle t_0, i_0, o_0 \rangle, x_0 \rangle, \ldots, \langle \langle t_k, i_k, o_k \rangle, x_k \rangle \}$}})
    {
      \scriptsize\textsf{
        \label{alg:receipt-valid:start}
        \ForEach{$\langle \langle t_i, i_i, o_i \rangle, x_i \rangle \in \mathcal{R}$}
        {
          \lIf {\KwNot \FVerifyReceipt{$\langle t_i, i_i, o_i \rangle, x_i$}}
          {
            \KwRet{invalidReceipt}
            \label{alg:receipt-valid:end}
          }
        }
    }}
    \KwOn({\FGetCheckpointAndLedger{$x_0, x_k$}})
    {
      \scriptsize\textsf{
        \For{$ C_0, s_{C_0}, \mathcal{L}, r \gets $\FEnforcerGetLedgerPackage{$x_o, x_k$}}
        {
          \label{alg:auditing-enforcer-get-ledger}
          $\mathsf{uPoM} \gets \mathsf{nil}$ \;
          \ForEach{$s \in s_{C_0}, ...,  \mathsf{seqno}(x_k+P)$}
          {
            \label{alg:auditing-request-checkpoint}
            \If{\KwNot \FIsBatchWellformed{${\mathcal{L}}, s$}}
            {
              \label{alg:auditing-request-is-wellformed}
              $\mathcal{F} \gets $\FCreateLedgerFragment{$\mathsf{nil}, s, \mathcal{L}$}\;
              $\mathsf{uPoM} \gets \langle \mathsf{nil}, \mathcal{F}, r \rangle$; $\mathsf{send(uPoM)}$; \KwRet{}\;
            }
          }
          \lIf{$\mathsf{uPoM} = \mathsf{nil}$}
          {
            \KwRet{$C_0, s_{C_0}, \mathcal{L}$}
          }
        }
    }}
    \KwOn({\FVerifyReceiptsInLedger{$\mathcal{R}, \mathcal{L}$}})
    {
      \scriptsize\textsf{
        \ForEach{$\langle \langle t_i, i_i, o_i \rangle, x_i = \langle v, s, H(k_p), \ldots {\mathcal{K}}_s, {\mathcal{S}} \rangle \rangle \in \mathcal{R} $}
        {
          \If {\KwNot $\mathsf{isReceiptInBatch(x_i, \mathcal{L})}$}
          {            
            \label{alg:auditing-request-is-receipt-in-batch}
            $\mathcal{F} \gets $\FCreateLedgerFragment{$\mathsf{nil}, s, \mathcal{L}$}\;
            $\mathsf{uPoM} \gets \langle \mathcal{F}, \langle \langle t_i, i_i, o_i \rangle, x_i \rangle \rangle$; $\mathsf{send(uPoM)}$; \KwRet{}\;
          }
        }
    }}
    \KwOn({\FReplayLedger{$C_0, s_{C_0}, \mathcal{L}$}})
    {
      \scriptsize\textsf{
        $s_{cp} \gets s_{C_0}$; $cp \gets C_0$; $\mathsf{kv} \gets \mathsf{loadCheckpoint}(s_{C_0}, C_0)$\;
        \label{alg:auditing-load-checkpoint}
        \ForEach{$s \in s_{C_0}, ...,  \mathsf{seqno}(x_k)$}
        {
          \label{alg:auditing-replay-start}
          \ForEach{$\langle t_i, i_i, o_i \rangle \in s$}
          {
            $\mathcal{L}, \mathsf{kv} \gets $\textsf{replayRequest}(${\mathcal{L}}, \mathsf{kv}, t_i$)\;
            \label{alg:auditing-replay}
            \If {\KwNot \FVerifyReplay{${\mathcal{L}}, \mathsf{kv}, \langle t_i, i_i, o_i \rangle$}}
            {            
              \label{alg:auditing-replay-in-batch}
              $\mathcal{F} \gets $\FCreateLedgerFragment{$s_{cp}, s, \mathcal{L}$}\;
              $\mathsf{uPoM} \gets \langle i_i, \mathcal{F}, cp \rangle$; $\mathsf{send(uPoM)}$; \KwRet{}\;
              \label{alg:upom}
            }
          }
          \If{$s \bmod C = 0$}
            {
              $s_{cp} \gets s$; $cp \gets \mathsf{createCheckpoint(kv)}$\;
            }
        }
    }}
\end{algorithm} 

If the batch at $s_{C_0}$ is not a checkpoint or the checkpoint digest does not match the first receipt, the auditor can assign blame to the intersection of replicas that have signed the batch at $s_{C_0} + C$ and the first receipt, as the checkpoint reference in a receipt must always link to the last committed checkpoint. If the fragment is not long enough to include the sequence number in one of the receipts, there must be misbehavior during a view change. The auditor can then blame at least $f$$+$$1$ misbehaving replicas: the intersection of the replicas that participated in a view change and that also signed the receipt. A correctness proof and the details of obtaining a complete ledger fragment and checkpoint are described in \Cref{sec:auditing-correctness}, Lemmas \ref{lemma:obtain-complete-ledger} and \ref{lemma:serializability-violations}.

After obtaining a well-formed ledger, in \textsf{verifyReceiptsInLedger} (line~\ref{alg:auditing-verify-receipts}), the auditor compares the receipts with the ledger. If a receipt $\langle \langle t_k, i_k, o_k \rangle, x_k \rangle$ does not match the batch at $s_k$ in the ledger fragment, we show in \Cref{lemma:incompatibility} that the auditor can assign blame to $f$$+$$1$ misbehaving replicas. In summary, there are three cases: (i)~the \msg{pre-prepare} with sequence number~$s_k$ in $\mathcal{L}$ has a view number $v_l = v_k$; (ii)~$v_l > v_k$; or (iii)~$v_l < v_k$. In case~(i), the ledger fragment contains evidence that the batch with sequence number~$s_k$ has prepared at $N$$-$$f$ replicas. Since at least $f$$+$$1$ of the replicas that have prepared the batch also signed the receipt, they can be blamed. In case~(ii), since $v_l > v_k$, there must be at least $N$$-$$f$ \msg{view-change} messages from different replicas that transition to a view greater than $v_k$ in the ledger fragment but claim not to have prepared the batch in the receipt in view~$v_k$. Since there are at least $f$$+$$1$ of those replicas that also signed the receipt, they can be blamed. In case~(iii), since $v_k > v_l$ and the ledger fragment is complete in relation to the receipt, there must be at least $N$$-$$f$ \msg{view-change} messages from different replicas that transition to a view greater than $v_l$ in the ledger fragment. Similarly, the intersection of those replicas and the ones that signed the receipt can be blamed. 

Since $N$$-$$f$ or more replicas may have misbehaved, it is necessary to replay transaction execution to check if the results are correct. The auditor does not need to understand the semantics of the service; it can retrieve the code of the stored procedures from $C_0$. The auditor sets the service state to the checkpoint value and replays transactions. If replaying a transaction fails to match the result in the ledger, the auditor can assign blame to any replica that signed the batch that contains the transaction. This is shown in \textsf{replayLedger} (line~\ref{alg:auditing-replay-no-error}).

\subsection{Enforcement}
\label{subsec:upom}

Since \sys provides individual accountability even if all replicas and members misbehave, there must be an \emph{enforcer} outside of the system to obtain checkpoints and ledger fragments for auditing, and to punish members responsible for misbehaving replicas. For example, consortium members may sign a binding contract to establish penalties if a uPoM proves that one of their replicas misbehaved, or if they fail to produce checkpoints and ledgers for auditing by an agreed deadline. These penalties may be imposed by the enforcer via arbitration~\cite{arbitration} or a court of law~\cite{discovery}.

The enforcer receives a set of receipts $\mathcal{R}$ from the auditor~(\Cref{alg:auditing-with-receipts}, line~\ref{alg:auditing-enforcer-get-ledger}). It then verifies that the receipts are valid, and requests all of the replicas that signed the latest receipt with the highest view for a ledger fragment that is complete in relation to $\mathcal{R}$.

Correct replicas will respond to the enforcer quickly. If the enforcer does not receive a response from a replica within a reasonable duration, \eg within minutes, it contacts the controlling consortium member to obtain the checkpoint and ledger. If the member fails to provide this information by an agreed deadline, \eg within days, it is punished according to the contract. This is important to ensure that misbehaving members cannot escape punishment by failing to produce information for auditing. However, it introduces a weak synchrony assumption that may lead to the punishment of honest but slow members. We expect that the deadline will be chosen conservatively to make this unlikely in practice. After the deadline elapses, the enforcer either returns to the auditor~$f$$+$$1$ responses, or it penalizes $f$$+$$1$ unresponsive replicas.

The enforcer also punishes members if a uPoM proves that one of their replicas misbehaved. When it receives a uPoM,
it checks its validity by carrying out an audit, as described in \S\ref{sec:auditing:process}, but the ledger fragment size and the number of transactions to replay is bounded by the transactions between two consecutive checkpoints. Furthermore, if there are fewer than $N-f$ misbehaving replicas, the uPoM does not require the enforcer to replay transactions. If the uPoM is incorrect, the enforcer punishes the auditor; otherwise, it punishes the members responsible for at least $f$$+$$1$ misbehaving replicas.

In practice, we expect the load placed on the enforcer to be small, because auditing is rare---\sys provides linearizability with up to $f$~misbehaving replicas and the enforcer penalizes entities that request information for auditing and fail to produce a valid, minimal uPoM.


\section{Reconfiguration and auditing}
\label{sec:governance}

In this section, we describe how \sys can change the consortium membership and the active replica set~(\S\ref{sec:governance:reconfiguration}). We explain how this impacts receipt validation~(\S\ref{sec:governance:receipts_auditing}) and auditing (\S\ref{sec:governance:auditing}).

\subsection{Reconfiguration}
\label{sec:governance:reconfiguration}

An \sys deployment must handle changes to the active member and replica set while supporting auditing, regardless of how many replicas misbehave. For this, \sys maintains governance data in the form of a \emph{configuration}, which includes the public signing keys for members and replicas and an endorsement of each replica's signing key signed by the member responsible.

Changing the configuration enables members to change the active replica set. This is initiated by a \emph{referendum}: members propose an updated configuration followed by the other members voting on the proposal. The number of votes required to pass the proposal is part of the service's state.

When voting on proposals, members must ensure the integrity of the service, \eg disallowing an individual member from controlling too many replicas. Members are also limited to adding or removing at most $f$~replicas, which ensures that the configuration change does not effect the service's liveness.

A referendum is carried out through governance transactions: a member proposes a new configuration by sending a \emph{propose} transaction request. This is followed by members sending \emph{vote} requests. Upon executing the final \emph{vote} transaction required for a referendum to pass at sequence number $s$, the primary ends the current batch, and initiates the reconfiguration process.

A \emph{reconfiguration} first adds evidence for the referendum to the ledger. This is done as part of the old configuration by the primary sending $P$~\msg{pre-prepare} messages without batched requests, called the \emph{end-of-configuration} batches. The \msg{pre-prepare} message for the end-of-configuration batch at sequence number~$s+P$ contains evidence that the batch at $s$ committed (\S\ref{sec:receipts}). In addition, these \msg{pre-prepare} messages include an extra field: the \emph{committed} Merkle root, which is the root of the Merkle tree at $s$.
This evidence is required for auditing: it commits the replicas that signed the $P^{\mathrm{th}}$ end-of-configuration batch to triggering the reconfiguration.
Similarly, the signatures of the replicas that prepared the $P^{\mathrm{th}}$ end-of-configuration batch must be included in the ledger in the same configuration. Following the first $P$ end-of-configuration batches, the primary pre-prepares another set of $P$ end-of-configuration batches. The configuration change takes effect at $s$$+$$2P$.

The replicas in the new configuration create a checkpoint of the key-value store at sequence number~$s$$+$$2P$. The primary creates a \msg{pre-prepare} for the checkpoint at $s$$+$$2P$$+$$1$, followed by $P$ \emph{start-of-configuration} \msg{pre-prepare} messages with empty request batches. This ensures that a correct replica commits the checkpoint transaction before other transactions are executed in the new configuration. If any of the end/start-of-configuration batches correspond to a checkpoint sequence number, the checkpoint is skipped. Therefore, the checkpoint digests~$d_c$ in the \msg{pre-prepare} messages always refer to checkpoints in the same configuration.

A newly added replica first obtains the ledger and a recent checkpoint, and replays the ledger from that checkpoint~(\S\ref{sec:receipts:optimisations}). Replicas that are no longer part of the new configuration retire after sending the pre-prepare for~$s$$+$$2P$. Removed members and replicas should delete their private signing keys to provide forward security. This prevents them from being blamed for future compromises, while still allowing authentication of transactions in the ledger using their public keys.

\subsection{Governance sub-ledger and receipts}
\label{sec:governance:receipts_auditing}

When a client verifies a receipt, it must know which replicas were active when the receipt was created. \sys addresses this with the help of the governance sub-ledger.

Governance transactions are recorded in the ledger and used by auditors to determine the active configuration. Clients, however, do not have a copy of the ledger, but need to verify receipt signatures. To do this, they store receipts for all governance transactions and, for each reconfiguration, they also store the receipts for the $P^{\mathrm{th}}$ end-of-configuration batch. We refer to this as the receipts of the \emph{governance sub-ledger}. A client checks that a transaction receipt for index~$i$ is valid by considering the governance sub-ledger from the genesis transaction~$\mathit{gt}$ up to $i$. The client verifies the governance receipts, and if successful, the replica signing keys at index~$i$ are used to validate the receipt~(\S\ref{sec:receipts}).

This raises the challenge of how a client determines that it has \emph{all} required governance receipts. \sys includes the ledger index of the last governance transaction in each \msg{pre-prepare} message and receipt ($i_g$). A client can request missing receipts from replicas by traversing the sequence of governance receipts. It verifies received receipts incrementally and caches them locally.

With reconfiguration, the definition of a valid receipt is extended: a valid receipt~$R$ must include valid governance receipts from $\mathit{gt}$ up to the configuration that produced $R$.

\subsection{Auditing}
\label{sec:governance:auditing}

Reconfiguration introduces several new tasks for the auditor: it must consider the governance sub-ledger with receipts; validate that reconfigurations were executed correctly; and ensure that that only one configuration was active for any given index or sequence number. Next, we provide a summary of the required changes to the auditing process; a detailed correctness proof is included in \Cref{sec:auditing-correctness-proof-reconfig}.

A client initiates an audit by sending inconsistent receipts and the supporting governance receipts to an auditor. The auditor replays these governance transactions to determine the signing keys required to verify each client receipt. After verifying the receipts, the auditor requests a ledger fragment and checkpoint from the enforcer.

The auditor may uncover that multiple configurations were active for a given index or sequence number, this can happen when misbehaving replicas fork or rewrite the ledger. We call this a \emph{fork in governance}. If the auditor finds a fork, there are two $P^{\mathrm{th}}$ end-of-configuration batch receipts with the same preceding configuration that are not \emph{equivalent}: they are at different indices or sequence numbers, or their \msg{pre-prepare} messages do not contain the same committed Merkle root, \ie they are not preceded by the same governance transactions. In this case, the auditor assigns blame to the replicas that signed both receipts, as a correct replica that prepares a $P^{\mathrm{th}}$ end-of-configuration batch commits the final \emph{vote} transaction that triggers reconfiguration.

If the enforcer cannot obtain the required information for a valid receipt~$R$ from the sequence of provided receipts, there must be misbehaving replicas. In addition to the misbehavior described in \S\ref{sec:auditing:process}, the misbehaving replicas may have created a fork in governance or incorrectly prepared the $P^{\mathrm{th}}$ end-of-configuration batch that succeeds the configuration that produced the receipt~$R$ (see Lemmas~\ref{lemma:obtain-complete-ledger-reconfig} and \ref{lemma:serializability-violations-reconfig}).

Another possibility is that the configuration that produced a receipt~$R$ for a sequence number~$s$ may not match the configuration that prepared the batch at $s$ in a well-formed ledger fragment. In this case, blame is again assigned to the replicas that signed $R$ and prepared the $P^{\mathrm{th}}$ end-of-configuration batch that succeeds the configuration that produced $R$ (see Lemma~\ref{lemma:config-mismatch}). 

After assigning blame, the auditor sends a uPoM to the enforcer with the supporting governance receipts.


\section{Evaluation}
\label{sec:eval}

\newcommand{\numbersFOneThroughput}{47,841}
\newcommand{\numbersFOneNoReceiptThroughput}{51,209}

We evaluate \sys to understand the cost of providing receipts~(\S\ref{sec:evaluation:throughput-lat}), its scalability~(\S\ref{sec:eval:scalability}), the overheads of receipt validation~(\S\ref{sec:eval:receipt_validation}), and auditing~(\S\ref{sec:eval:auditing}). We finish with a performance breakdown of \sys{}'s design features~(\S\ref{sec:eval:breakdown}).

\mypar{Testbeds} Our experimental setup consists of three environments: (a)~a dedicated cluster with 16~machines, each with an 8-core 3.7-Ghz Intel E-2288G CPU with 16\unit{GB} of RAM and a 40\unit{Gbps} network with full bi-section bandwidth; (b)~a LAN environment in the Azure cloud, with Fsv2-series VMs with 16-core 2.7-GHz Intel Xeon 8168 CPUs and 7\unit{Gbps} network links; and (c)~a WAN environment with the same VMs across 3 Azure regions (US East, US West 2, US South Central). All machines run Ubuntu Linux~18.04.4 LTS.

\mypar{Implementation} Our \sys prototype is based on CCF~v0.13.2~\cite{Releasec85:online} and has approx. 40,000~lines of C++ code. It uses the formally-verified Merkle trees and SHA functions of EverCrypt~\cite{protzenko2020evercrypt}, the MbedTLS library~\cite{embedtls} for client connections, and secp256k1~\cite{wuille2018libsecp256k1} for all secure signatures. Replicas create secure communication channels using a Diffie–Hellman key exchange.

Pipelining batch execution ($P$ in \Cref{alg:lpbft}) improves \sys{}'s throughput. We use $P$$=$$2$ for the LAN and $P$$=$$6$  for the WAN, with maximum batch sizes of 300 and 800~requests, respectively. Checkpoints are created every 10K or 4K sequence numbers in the LAN and WAN environments, respectively.

\begin{table}[tb]
  \centering
  \caption{Size of ledger entries \textnormal{(SmallBank)}}
  \label{tbl:ledger-entry}\small\renewcommand{\arraystretch}{0.9}
  \begin{tabular}{lrr}
    \toprule
    \multirow{2}{*}{Ledger entry type} &   \multicolumn{2}{c}{Size (bytes)}  \\
                      &  \multicolumn{1}{c}{f = 1} & \multicolumn{1}{c}{f = 3}  \\
    \midrule
    Transaction (SmallBank) & \multicolumn{2}{c}{216--358} \\
    \msg{Pre-prepare} & \multicolumn{2}{c}{277} \\
    \msg{Prepare} Evidence & 298 & 894 \\
    Nonces & 32 & 64 \\
    \bottomrule
  \end{tabular}
\end{table}

\mypar{Benchmarks} We use the \emph{SmallBank} benchmark~\cite{smallbank}, which models a bank with 500K customer accounts. Clients randomly execute 5~transaction types: deposit, transfer, and withdraw funds; check account balances; and amalgamate accounts. The size of the ledger entries is shown in \Cref{tbl:ledger-entry} where only the Prepare Evidence and Nonces entries depend on $f$.

Since \sys{}'s design targets accountability with more than $f$ failures, we omit results from experiments with fewer failures. In such cases, \sys{}'s performance matches that of prior work, because it uses well-established BFT techniques, such as view changes, sending messages via authenticated channels and client-signed requests~\cite{bessani2014state, aardvark}. Instead, we consider the performance of receipt validation~(\S\ref{sec:eval:receipt_validation}) and auditing~(\S\ref{sec:eval:auditing}), which are new contributions of \sys.

Transaction throughput is measured at the primary replica and latency at the clients. All experiments are compute-bound. Results are averaged over 5~runs, with min/max error bars.

\mypar{Baselines} We compare against four baselines: \textsf{\sys-PeerReview}, which uses PeerReview for accountability~\cite{peerreview}, i.e., replicas sign all messages and send signed acknowledgements for all messages; \textsf{\sys-NoReceipt}, an \sys variant that produces a ledger but no receipts; \textsf{HotStuff}~\cite{hotstuff}, a state-of-the-art BFT protocol, which is at the core of the Diem permissioned ledger system~\cite{libra}; and Hyperledger~\textsf{Fabric} (v.~2.2)~\cite{hyperledger_fabric}, a popular open-source permissioned ledger system. We compare against Fabric's latest major release that does not include a BFT consensus protocol~\cite{hyperledger_fabric_bft} and only tolerates crash failures using Raft~\cite{ongaro2014raft}.

\begin{figure}[tb]
  \centering
  \includegraphics[width=1.0\columnwidth]{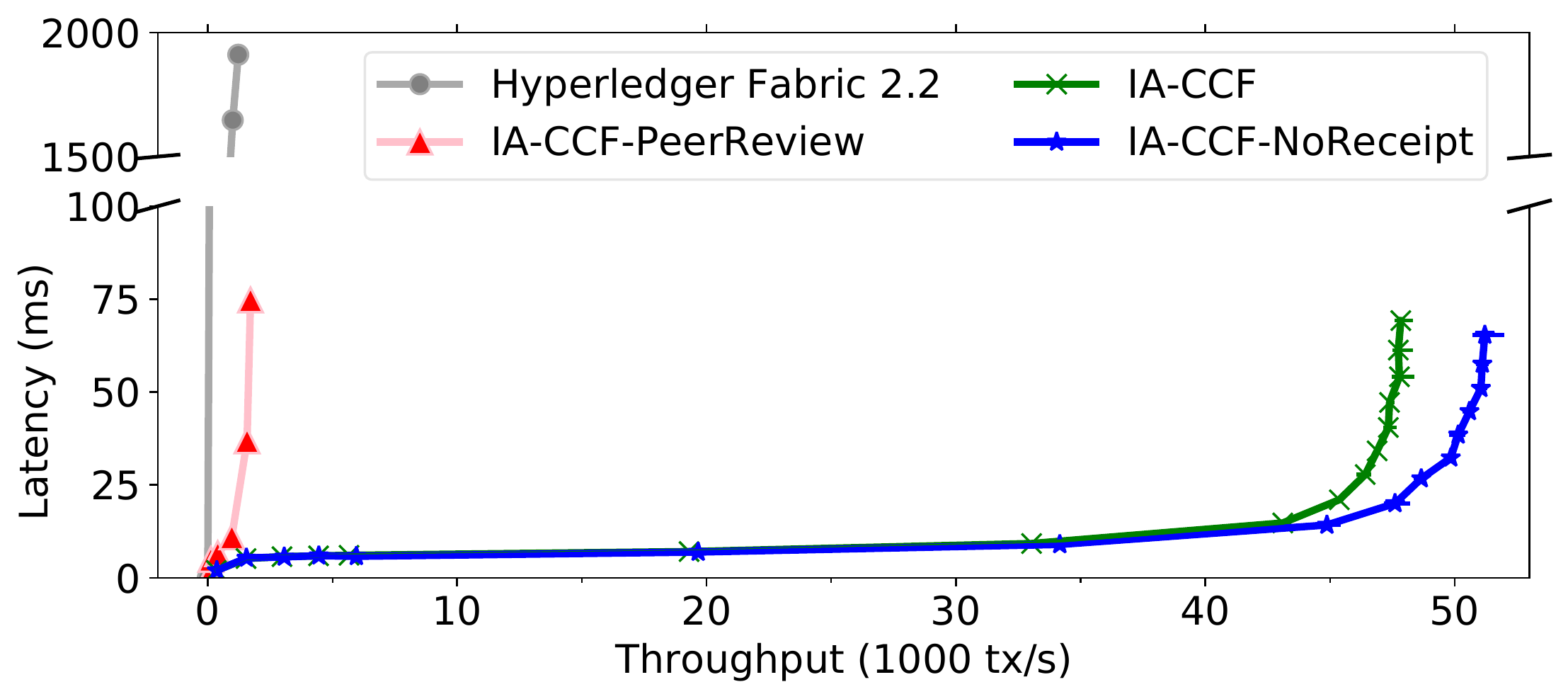}
  \caption{Transaction throughput/latency \textnormal{($f$$=$$1$, dedicated cluster)}}
  \label{fig:throughput-latency}
\end{figure}

\subsection{Transaction throughput and latency}
\label{sec:evaluation:throughput-lat}

We explore the throughput and latency of transaction execution with 4~replicas ($f$$=$$1$) in the dedicated cluster, comparing \textsf{\sys}, \textsf{\sys-NoReceipt}, \textsf{\sys-PeerReview}, and \textsf{Fabric}.

\Cref{fig:throughput-latency} shows a throughput/latency plot as transaction load increases. \textsf{\sys} achieves \numbersFOneThroughput\unit{tx/s} while maintaining latencies below 70\unit{ms}. As the load increases, queueing delays increase latency. \textsf{\sys-NoReceipt}'s throughput is \numbersFOneNoReceiptThroughput\unit{tx/s}, which is only 3\% higher than \sys, demonstrating the low cost of receipts.

\textsf{\sys-PeerReview} exhibits an order of magnitude lower throughput because all messages must be signed, e.g., a replica must sign a reply message for each transaction in a batch.
 This causes \textsf{\sys-PeerReview} to perform two orders of magnitude more asymmetric cryptographic operations than \sys.

\textsf{Fabric}'s throughput is only 1,222\unit{tx/s}, with a latency of 1.9\unit{s}. This is substantially worse than \sys, despite not using a BFT protocol. Our analysis reveals two reasons: Fabric's \emph{execute-order-validate} model requires that replicas issue a signature for each executed transaction, while \sys replicas only require one signature per batch; and Fabric suffers from documented inefficiencies related to its key-value store implementation~\cite{nakaike2020hyperledger}.

\begin{table}[tb]
  \centering
  \caption{Request latency under low load \textnormal{(WAN)}}
  \label{tbl:two-round-trips}\small\renewcommand{\arraystretch}{0.8}
  \begin{tabular}{lrrr}
    \toprule
  & average & 99\textsuperscript{th} percentile & network  \\
  & latency & latency  & round trips \\
    
    \midrule
    \textsf{\sys} & 183\unit{ms} & 194\unit{ms} & 2 \\
    \textsf{HotStuff} & 340\unit{ms} & 393\unit{ms} & 4.5 \\
    \bottomrule
  \end{tabular}
\end{table}

\begin{figure}[tb]
  \centering
  \includegraphics[width=1.0\columnwidth]{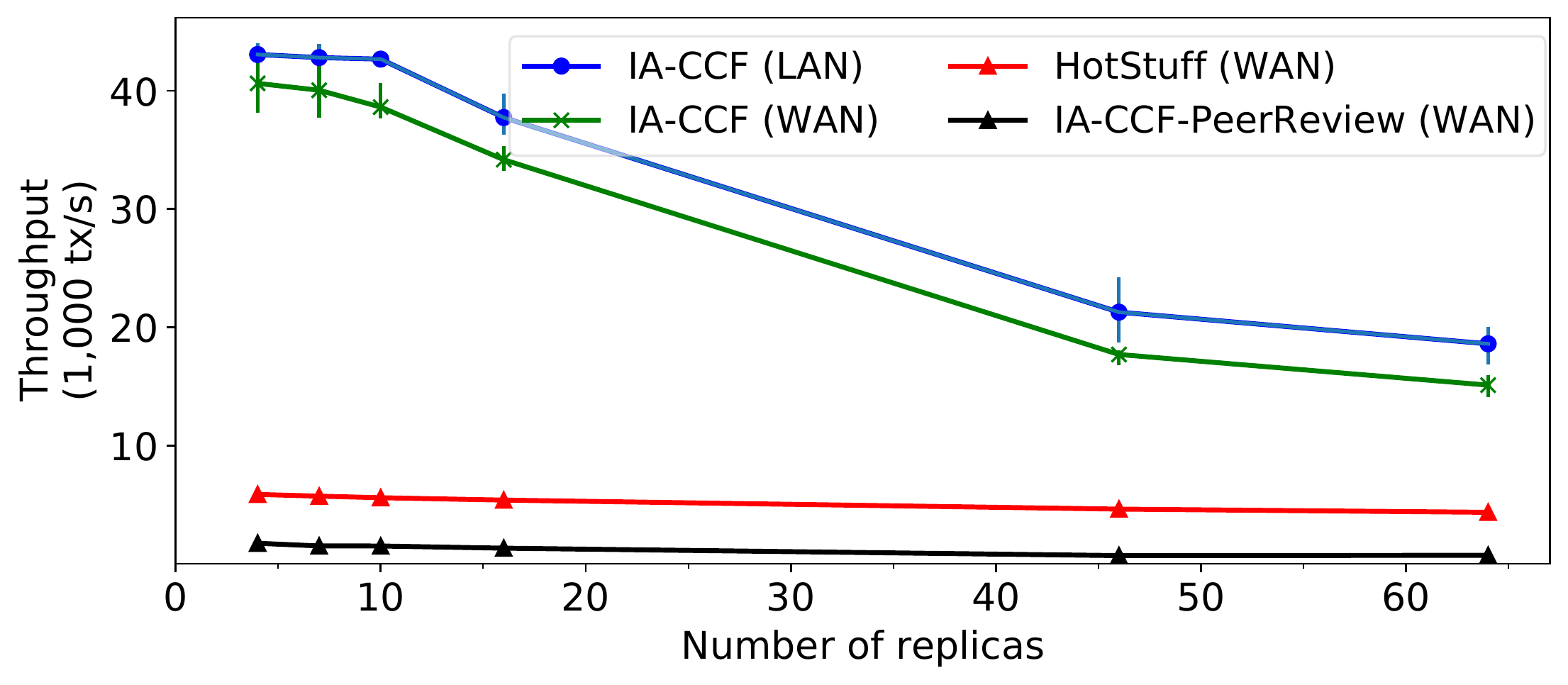}
  \caption{Transaction throughput vs. replica count \textnormal{(WAN)}}
  \label{fig:scale-out-eval}
\end{figure}

\begin{figure*}[t]
  \centering
  \begin{subfigure}[t]{.33\linewidth}
    \centering
    \includegraphics[width=1.0\columnwidth]{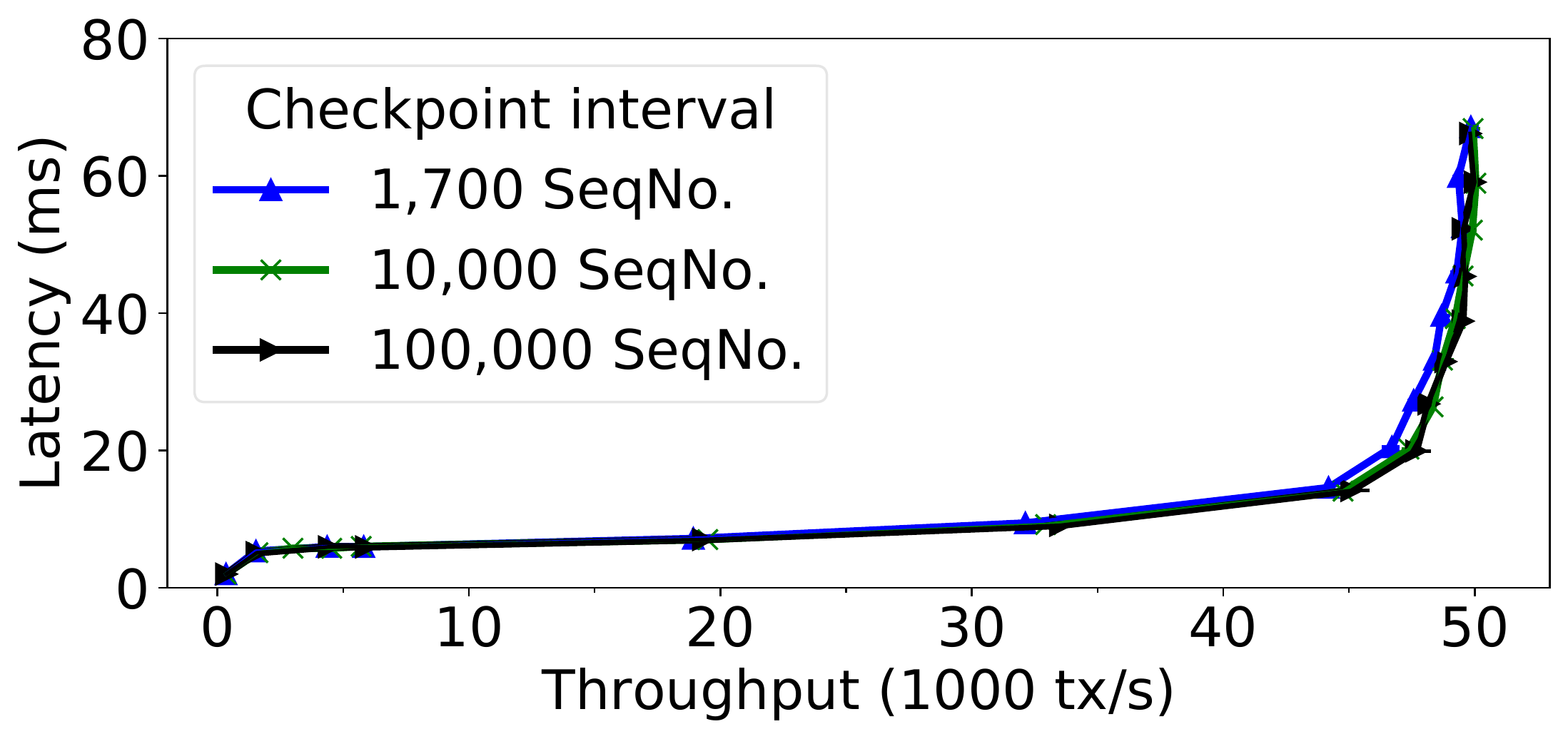}
    \caption{100,000 accounts}
    \label{fig:throughput-checkpoint-u100000}
  \end{subfigure}
  \begin{subfigure}[t]{.33\linewidth}
    \centering
    \includegraphics[width=1.0\columnwidth]{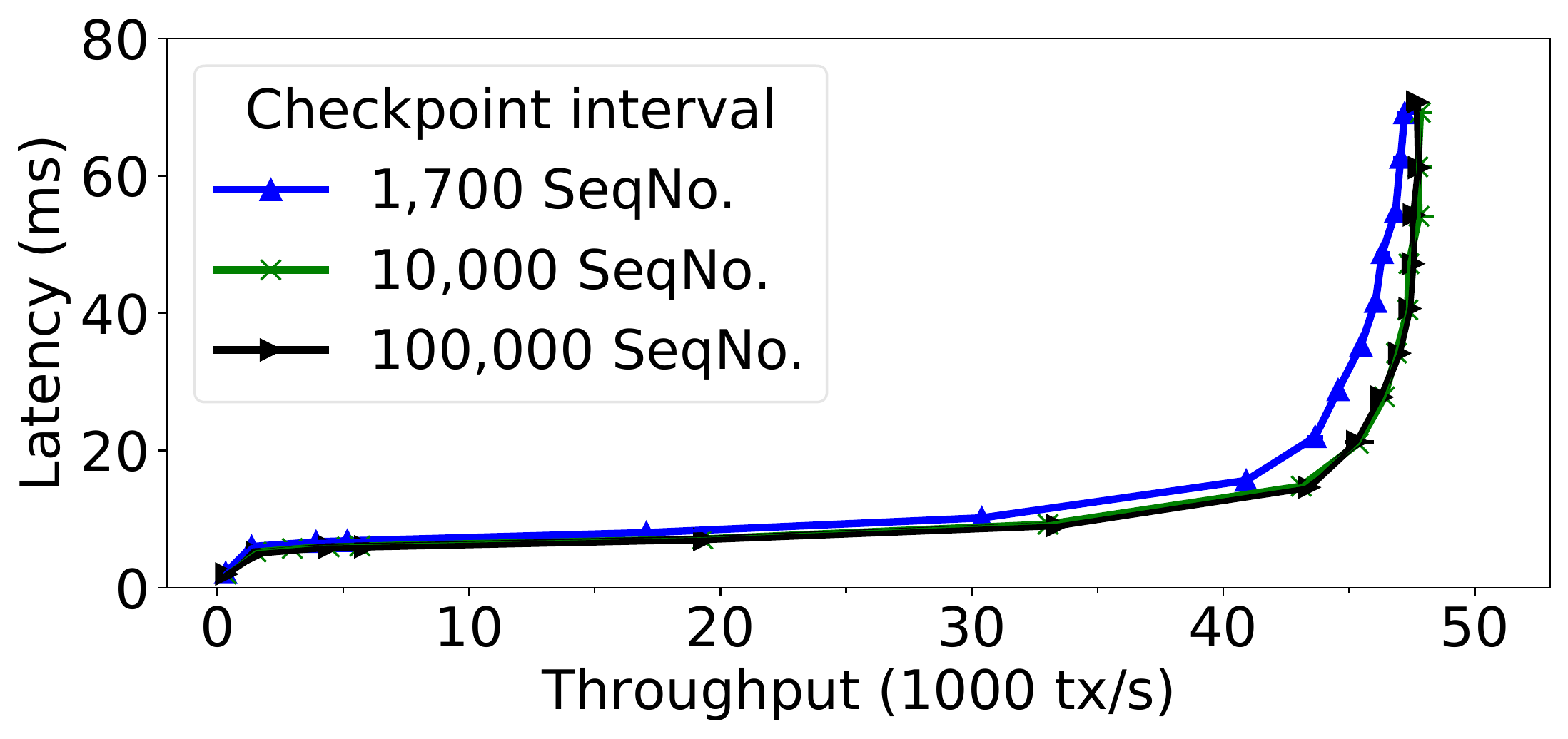}
    \caption{500,000 accounts}
    \label{fig:throughput-checkpoint-u500000}
  \end{subfigure}
  \begin{subfigure}[t]{.33\linewidth}
    \centering
    \includegraphics[width=1.0\columnwidth]{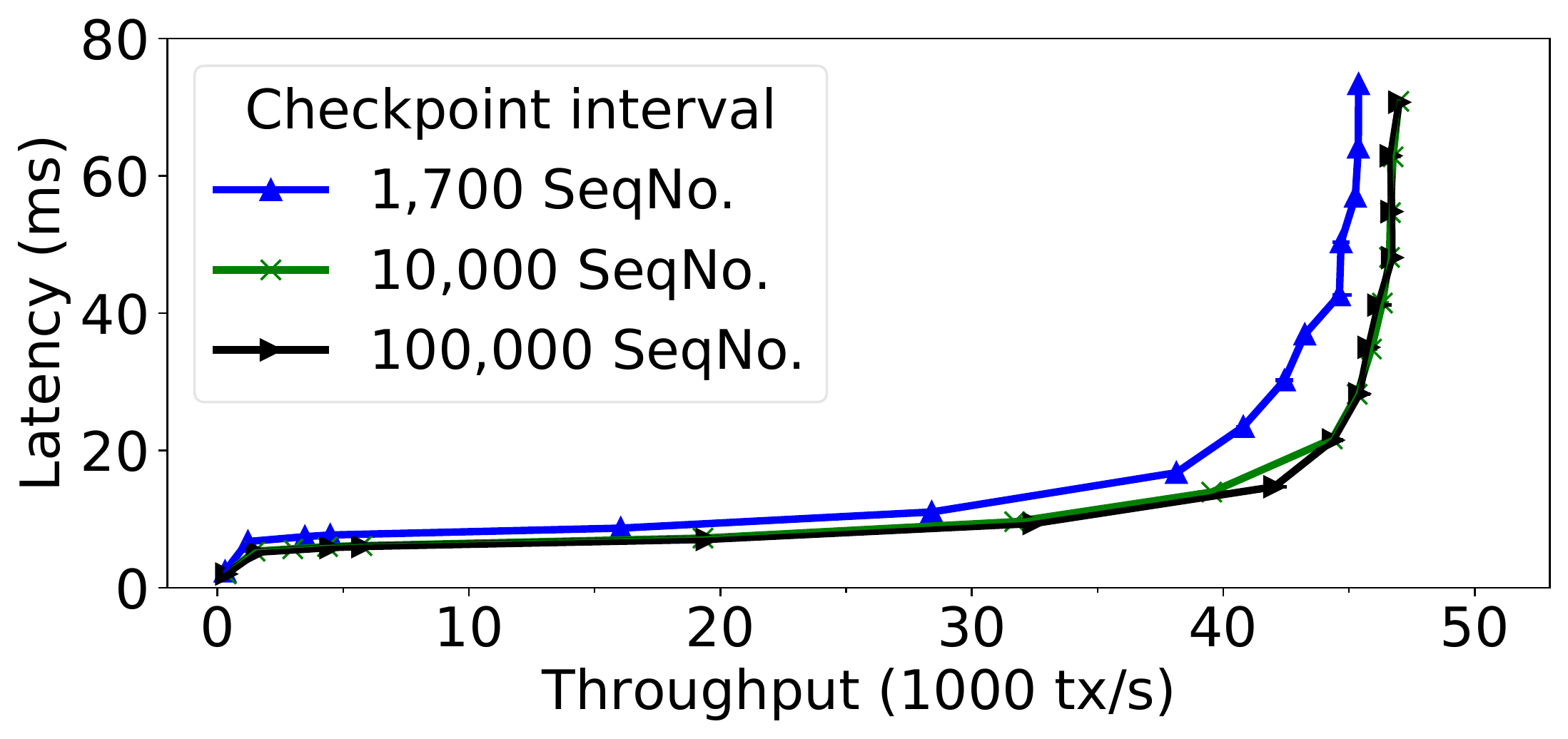}
    \caption{1,000,000 accounts}
    \label{fig:throughput-checkpoint-u1000000}
  \end{subfigure}
  \caption{\label{fig:throughput-checkpoint}Transaction throughput/latency when varying the number of accounts and checkpoint interval \textnormal{($f$$=$$1$, dedicated cluster)}}
\end{figure*}
\subsection{Scalability}
\label{sec:eval:scalability}

Next we consider the effect on transaction throughput when increasing the number of \sys replicas in the Azure WAN environment, spanning multiple regions to reduce correlated failures~\cite{barr2011building}. We compare against \sys deployed in the Azure LAN environment, \textsf{\sys-PeerReview}, and \textsf{HotStuff}, a BFT consensus protocol without a ledger or key-value store.

\Cref{fig:scale-out-eval} shows that, as expected, \textsf{\sys}'s throughput decreases with more replicas because more signatures are verified by each replica. Since each replica has a fixed number of threads for checking signed \msg{pre-prepare}/\msg{prepare} messages in parallel, throughput decreases when the replica count exceeds the number of hardware threads, which is only 16 in this deployment. 
\textsf{\sys} is only marginally affected by the higher WAN latencies due to its use of pipelining, as shown by the comparison to the LAN deployment.

\textsf{HotStuff}~\cite{hotstuff-commit} achieves a throughput of 5,862\unit{tx/s} in the WAN environment, which is worse than its reported LAN throughput~\cite{pompe-commit}. While it degrades slowly with more replicas, even with 64~replicas its throughput remains 71\% lower than that of \textsf{\sys}. 
The throughput of \textsf{\sys-PeerReview} is even lower since it performs more cryptographic operations.

We also measure the request latency of  \textsf{HotStuff} and \sys under low load. As reported in \Cref{tbl:two-round-trips}, \textsf{HotStuff}'s request latency is approximately twice that of \textsf{\sys}'s. For both systems, request latency is dominated by the number of network round trips and clients receive transaction results with receipts in only 2 round trips in \sys. 

\label{sec:eval:req_lat}

\subsection{Receipt validation}
\label{sec:eval:receipt_validation}

We measure the time required to verify receipts, which depends on (i)~the length of the path in the Merkle tree~$G$ and (ii)~the number of signatures to be checked. Since the number of leaves in $G$ is bounded by the batch size, the path length remains small: verification takes 2.1\unit{\textmu{}s} and 2.3\unit{\textmu{}s} for batches of 300 and 800 requests, respectively. The overall cost is dominated by the signature verification, which takes 18\unit{ms} and 52\unit{ms} for $f$$=$$1$ and $f$$=$$3$, respectively. 

\subsection{Governance sub-ledger}
\label{sec:eval:gov_subledger}

Next, we consider the size of the governance sub-ledger, which is stored by clients. The sub-ledger is a collection of receipts for every transaction that has updated the governance of an \sys deployment. A receipt's size is 623\unit{bytes} or 1,565\unit{bytes} for $f$$=$$1$ or $f$$=$$3$, respectively. In addition, the client must store the governance request and the corresponding response, which have variable size. We expect governance operations to be rare. Therefore, storing and verifying governance sub-ledger receipts has low overhead.

\subsection{Ledger auditing}
\label{sec:eval:auditing}

Next, we want to understand auditing performance. For the SmallBank workload, we compare execution time to auditing time. When measuring throughput at $f$$=$$1$, auditing is 23\% faster than execution, because there is no network overhead, message signing, or ledger writes. In each batch, \sys only verifies $2$$f$$+$$1$ rather than up to $3$$f$$+$$1$ signatures. For $f$$=$$4$, the performance gap increases to 67\%, as more replicas add communication and cryptographic load during execution. We observe that the bottleneck for
auditing is verification of client request signatures, which can be trivially parallelized.

\subsection{Key-value store}
\label{sec:eval:kv-size}

\begin{figure}[tb]
  \centering
  \includegraphics[width=1.0\columnwidth]{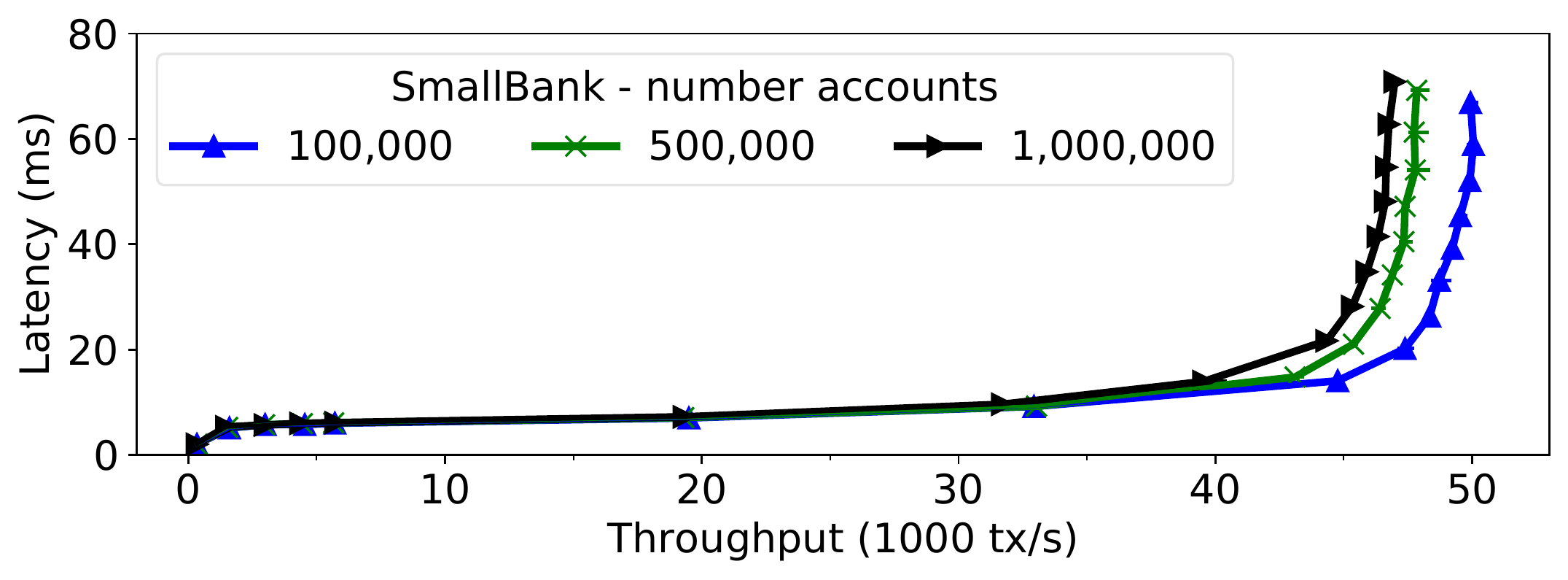}
  \caption{Transaction throughput/latency with different account numbers \textnormal{($f$$=$$1$, dedicated cluster)}}
  \label{fig:throughput-latency-db-size}
\end{figure}

We explore the performance impact of varying the number of entries in the key-value store by varying the number of SmallBank accounts.
\Cref{fig:throughput-latency-db-size} shows a throughput vs. latency plot. As expected, throughput decreases when the number of entries in the key-value store increases. CCF's implementation~\cite{russinovich2019ccf} of the key-value store uses a CHAMP map~\cite{steindorfer2016fast}, whose access time grows logarithmically with the number of items. 

\subsection{Checkpointing}
\label{sec:eval:Checkpointing}

We also explore the effect of checkpointing on performance. We vary the size of the key-value store and the checkpoint interval for the SmallBank workload. \Cref{fig:throughput-checkpoint} shows the results as throughput vs. latency plots. 
As expected, the checkpoint overhead increases with the size of the key-value store and the checkpoint frequency, but the overhead is low for checkpoint intervals between 10 and 100K (approximately 1 to 10 minutes).
The checkpoint interval impacts the overhead to check uPoMs at the enforcer.
We expect checkpointing every 10 minutes to be acceptable in practice; it requires the enforcer to replay at most 10 minutes of transactions.

\subsection{Overhead breakdown}
\label{sec:eval:breakdown}

To provide a permissioned ledger with individual accountability, \sys implements functionality that goes beyond traditional BFT consensus protocols, e.g., generating receipts. We now explore the impact of implementing this functionality on \sys{}'s throughput in the dedicated cluster.

We compare several variants of \sys, each limiting functionality further: (a)~\textsf{\sys}; (b)~\textsf{\sys-NoReceipt}, \ie without creating receipts; (c)~without creating checkpoints; (d)~with a small key-value store, \ie the key-value store fits in the CPU cache; (e)~without signed client requests; (f)~using only MACs for message authentication between replicas; (g)~without a ledger; and (h)~with empty requests, \ie without the overhead of executing transactions against the key-value store.

\begin{table}[tb]
\centering
    \caption{Breakdown of \sys features \textnormal{($f$$=$$1$, dedicated cluster)}}
    \label{tbl:feature-overhead}\small\renewcommand{\arraystretch}{0.8}
  \begin{tabular}{l@{\hspace{-10pt}}lrr}
      \toprule
      Variant && Throughput (tx/s) \\
      \midrule
      (a) & Full \textsf{\sys} & \numbersFOneThroughput \\
      (b) & \textsf{\sys-NoReceipt} & \numbersFOneNoReceiptThroughput \\
      (c) & \hspace{0.25em} + without checkpoints & 51,288 \\
      (d) & \hspace{0.25em} + small key-value store & 53,759 \\
      (e) & \hspace{0.25em} + without signed client requests & 111,926 \\
      (f) & \hspace{0.25em} + with MACs only & 128,921 \\
      (g) & \hspace{0.25em} + without ledger & 131,959 \\
      (h) & \hspace{0.25em} + with empty requests & 299,321 \\
      \midrule
       & \textsf{HotStuff} (with empty requests) & 307,997 \\
       & \textsf{Pomp\=e} (with empty requests) & 465,646 \\
      \bottomrule
    \end{tabular}
\end{table}

\Cref{tbl:feature-overhead} shows that (a)--(d) have comparable throughput, but not verifying client signatures~(e) doubles throughput. 
Only using MACs instead of signatures~(f) or removing the ledger altogether~(g) does not increase throughput substantially, but 
removing the overhead of executing transactions against the key-value store~(h) again doubles throughput.

For context, we compare with two Byzantine consensus protocols with similar functionality to (h) above, \textsf{HotStuff}~\cite{hotstuff} and \textsf{Pomp\=e}~\cite{zhang2020byzantine, pompe-commit}. \textsf{HotStuff}'s throughput is 307,997\unit{tx/s}, but with higher latency~(\S\ref{sec:eval:req_lat}). By separating request ordering and consensus, \textsf{Pomp\=e} achieves a throughput of 465,646\unit{tx/s}, also with worse latency (\sys's 12\unit{ms} to Pomp\=e's 73\unit{ms}). \sys could utilize Pomp\=e's techniques for increased throughput by sacrificing its two round-trip latency.

These breakdown results show that \sys{}'s overhead comes primarily from the cryptographic operations required for verifying client requests, followed by the transactional key-value store, rather than the consensus protocol or the mechanisms specific to providing individual accountability. 


\section{Related work}
\label{sec:rel_work}

\mypar{Permissioned ledgers} Many permissioned ledger systems~\cite{besu:online, Quorum, libra,hyperledger_fabric} rely on BFT consensus protocols to order transactions. Hyperledger Besu~\cite{besu:online} and Quorum~\cite{Quorum} use variants of PBFT~\cite{moniz2020istanbul, saltini2019ibft}, which do not retain proof of a replica's operations, and therefore cannot assign blame. Diem~\cite{libra} uses the DiemBFT~\cite{librabft} consensus protocol, which is based on HotStuff~\cite{hotstuff} and also lacks accountability features.

The \sys prototype is built on top of CCF~\cite{russinovich2019ccf}, an open source~\cite{microsof87:online} distributed ledger framework deployed in the Azure cloud~\cite{Confiden79:online}, which utilizes trusted execution environments~(TEEs)~\cite{costan2016intel, kaplan2016amd} to harden replicas. Russinovich \etal~\cite{russinovich2019ccf} describe CCF's programming model, receipts, governance, and replication protocols. CCF produces hardware attestation reports for the code running on each replica and adds them to the ledger. The ledger is signed by the CCF service and in the process binds CCF's public key to the code and hardware platform. While CCF enables auditing and can recover a ledger when all replicas crash, it relies on the security of TEEs, and its auditing does not guarantee individual accountability.

\myparr{Byzantine consensus}~\cite{pbft, aardvark, kotla2007zyzzyva} distributes trust. Recent work on BFT protocols has focused on improving guarantees~\cite{crain2018dbft, miller2016honey, asayag2018fair} or performance for particular use cases~\cite{stathakopoulou2019mir, zhang2020byzantine}. SBFT~\cite{sbft} and HotStuff~\cite{hotstuff} scale to hundreds of replicas using threshold cryptography, which prevents blame assignment. For permissioned ledgers, scaling to many replicas without growing the consortium size does not improve trustworthiness, and consortia typically cannot grow arbitrarily. 

Other work has explored misbehavior and its impact on Byzantine consensus.
BFT2F~\cite{bft2f} formalizes safety and liveness guarantees after more than $f$~replicas are compromised. It provides PBFT's guarantees with up to $f$~failures and provides \emph{fork*} consistency with up to $2f$ failures. For permissioned ledgers, \emph{fork*} consistency is not sufficient, because it is susceptible to double-spending attacks.

Depot~\cite{depot} issues proofs-of-misbehavior after observing misbehavior, but it adopts eventual consistency, which is incompatible with permissioned ledgers. Pomp\=e~\cite{zhang2020byzantine} prevents dishonest primaries from controlling the ordering of requests.  It does not address scenarios in which there are more than $f$ dishonest replicas though.

\mypar{Accountability} PeerReview~\cite{peerreview} ensures that distributed nodes remain accountable for their actions. As shown in \S\ref{sec:evaluation:throughput-lat}, PeerReview incurs a high overhead when applied to a permissioned ledger. In contrast, \sys introduces mechanisms specific to BFT state machine replication, such as a shared ledger with a Merkle tree, to improve both regular transaction execution and auditing. 

Accountable virtual machines~\cite{avm} carries out auditing through \emph{spot checking} of checkpoints, but has the same performance overheads as PeerReview for ledgers. SNP~\cite{snp} is a networking-specific implementation of accountability, offering provenance for routing decisions. Such specializations improve performance in particular domains, but are not directly applicable to permissioned ledgers.

BAR~\cite{bar} and Prosecutor~\cite{zhang2021prosecutor} incentivize replicas to act honestly by having honest replicas penalize misbehavior. This weaker model allows BAR to tolerate more than $1/3$ faulty replicas, while Prosecutor uses these incentives to improve performance. If these incentives fail~\cite{KryptonR56:online}, however, replicas share the blame.

Accountability with more than $f$$+$$1$ misbehaving replicas has been discussed before~\cite{buchman2016tendermint, buterin2017casper, herlihy2016blockchains}. BFT Protocol Forensics~\cite{kannan2020bft} and Polygraph~\cite{polygraph} propose a ledger auditing mechanism, but assume that fewer than $N$$-$$f$ replicas misbehave. They also do not support changing replica sets. ZLB~\cite{ranchal2020blockchain} and Tendermint~\cite{buchman2016tendermint} support changes to the replica set but also assume 
that fewer than $N$$-$$f$ replicas misbehave.



\section{Conclusions}
\label{sec:concl}

In permissioned ledger systems, individual accountability is a strong disincentive for misbehavior. \sys provides the evidence required to prove that $f$$+$$1$ or more replicas misbehaved when clients observe safety violations (even if all replicas fail). It offers strong consistency and security properties while providing state-of-the-art performance compared to existing ledger systems with weaker security guarantees. \sys achieves this by integrating evidence collection for assigning blame with a novel ledger-based BFT consensus algorithm.


\label{sec:acknowledgements}

\mypar{Acknowledgements} We thank our shepherd, Xiaowei Yang, and the anonymous reviewers for their valuable feedback.


\bibliographystyle{plain}
\bibliography{paper}

\appendix

\renewcommand{\tinyskip}{\vspace{3pt}}


\section{Proof of \LPBFT linearizability}
\label{proof:lpbft}

We present a correctness proof for \LPBFT. In particular, we show that \emph{early execution} (\Cref{lemma:early-execution}) and the \emph{nonce commitment} scheme (\Cref{lemma:nonce-commitment}) are equivalent to their counterpart features in PBFT. In \Cref{lemma:linearizability}, we show linearizability of \LPBFT.

\begin{lemma}[Rollback] \label{lemma:rollback}
  Any honest \LPBFT replica can roll back a suffix of the sequence of previously executed transaction batches.
\end{lemma}

\begin{proof}
  \LPBFT's state is distributed across several entities: a key-value store~$\mathit{kv}$; a Merkle tree~$M$; a ledger~${\mathcal{L}}$; a set of requests waiting to be ordered ${\mathcal{T}}$; a message store~${\mathcal{M}}$; and a nonce store~$\mathcal{K}$. Therefore, to roll back a batch of transactions, it must be possible to roll back all of these entities.

  \mypar{Key-value store~$kv$} The key-value store maintains a roll back transaction log. This enables transactions to be rolled back at a single transaction granularity. Thus, the last executed batch of transactions can be rolled back.

  \mypar{Merkle tree~$M$} When a new node is added to \LPBFT's Merkle tree, it becomes the right-most leaf of the tree. The value of a node in the tree is never updated, and a node can only be deleted if it is the right-most node in the tree. Thus, during roll back, it is possible to remove the nodes from the right of the tree that represent the last batch of executed transactions (in reverse order).

  \mypar{Ledger~$\mathcal{L}$} The ledger is represented by a file written to the disk by each replica. \LPBFT stores the index of all entries written to the ledger. To roll back the last executed batch, a \LPBFT replica truncates the ledger file to just before the first entry of the batch.

  \mypar{Transaction store~$\mathcal{T}$} It is not necessary to undo changes to the transaction store. Transaction requests that are removed can be retransmitted by the client or other replicas if needed.

  \mypar{Message store~$\mathcal{M}$, nonce store~$\mathcal{K}$} All items in the transaction and nonce stores are indexed by sequence number and view. Since roll back occurs only during a view change, and each item is associated with a view, it is not necessary to modify the message and nonce stores, because honest replicas never send more than one item of a given type for the same sequence number and view. 

  \tinyskip

  Therefore, it is possible to roll back a suffix of the sequence of transaction batches executed by \LPBFT replicas.

\end{proof}

\begin{lemma}[Early execution] \label{lemma:early-execution}
  \LPBFT's early execution and PBFT execution agree on all committed transactions.
\end{lemma}

\begin{proof}
  In both PBFT and \LPBFT, the primary determines the order of request execution by ordering requests into batches and assigning numbers to batches in \msg{pre-prepare} messages.
  In PBFT, requests are executed after commit and clients only accept results after transactions commit.
  In \LPBFT, requests are executed earlier, before the request even prepare, but the replicas only reply to clients after they prepare the requests and clients wait for matching replies from $N$$-$$f$ replicas.
  This ensures that they only obtain the transaction results after they commit as in PBFT.

  As in PBFT, a faulty primary may cause requests for which pre-prepares are sent not to commit. \LPBFT deals with this case by rolling back early execution (see \Cref{lemma:rollback}).
\end{proof}

\begin{lemma}[Nonce commitment] \label{lemma:nonce-commitment}
  The nonce commitment scheme is equivalent to replicas signing commit messages.
\end{lemma}

\begin{proof}
  \LPBFT, like PBFT, signs \msg{pre-prepare} and \msg{prepare} messages. Unlike PBFT, \LPBFT does not sign \msg{commit} messages. Replicas sample a fresh random nonce for each \msg{pre-prepare} or \msg{prepare} message with sequence number~$s$ at view~$v$, and add a hash of this nonce to the signed payloads. Later in the protocol, replicas include the nonce in the \msg{commit} message, instead of an extra signature.

  We show that this provides the same standard cryptographic security as the signature scheme (namely, resistance to existential forgery against chosen-message attacks) as long as the cryptographic hash function is second pre-image resistant on random inputs. Since the addition of a nonce to the signed payloads is injective, a forgery of a \LPBFT authenticator for a \msg{pre-prepare} or \msg{prepare} message yields a forgery against the signature scheme. A forgery of an authenticator for a \msg{commit} message, \ie a value with the same hash as a fresh random nonce that has not yet been revealed, is a second pre-image collision.
\end{proof}

\begin{theorem} \label{lemma:linearizability}
  \LPBFT is linearizable.
\end{theorem}

\begin{proof}
  \LPBFT changes the PBFT algorithm by adding early execution and the nonce commitment scheme.
  Lemmas \ref{lemma:early-execution} and \ref{lemma:nonce-commitment} show that these preserve the behavior of PBFT. 
\end{proof}


\section{Proof of auditing correctness}
\label{sec:auditing-correctness}

\newcommand{\tio}{\text{tio}}
\newcommand{\rgl}{\rangle}
\newcommand{\lgl}{\langle}
\newcommand{\minindex}[1]{{\color{blue}{#1}}}
\newcommand{\gov}[1]{{\color{purple!50!black}{#1}}}
\newcommand{\revisit}[1]{{\color{red!50!black}{#1}}}
\newcommand{\E}{\mathcal{E}}
\newcommand{\B}{\mathcal{B}}
\newcommand{\C}{\mathcal{C}}
\newcommand{\CP}{\mathit{cp}}
\newcommand{\Fr}{\mathcal{F}}
\newcommand{\M}{\mathcal{E}}
\newcommand{\GL}{\mathcal{N}}
\newcommand{\G}{n}
\newcommand{\R}{\mathcal{R}}
\newcommand{\Sf}{\mathcal{U}}
\newcommand{\pack}{$\Fr, \Sf, \M, \GL, \CP$}
\newcommand{\vc}{\emph{view-change }}
\newcommand{\nv}{\emph{new-view }}
\newcommand{\res}{\textsf{result }}
\newcommand{\vote}{\textsf{vote }}
\newcommand{\nlparagraph}[1]{\paragraph{#1}\mbox{}\\}
\newcommand{\mysubsection}[1]{\subsubsection{#1}}
\newcommand{\mysection}[1]{\subsection{#1}}

\newcommand{\slow}{\text{or slow}}

First, we present the correctness proof for auditing without governance transactions and reconfiguration (\S\ref{sec:auditing-correctness-proof}). Then, we extend the proof to include governance transactions and reconfiguration (\S\ref{sec:auditing-correctness-proof-reconfig}).

\mysection{Correctness of auditing without reconfiguration}
\label{sec:auditing-correctness-proof}

We begin with a description of terminology and notation. In \S\ref{sec:obtain-complete-ledger} and Lemma~\ref{lemma:obtain-complete-ledger}, we then prove that, given a set of receipts, the auditor, with the help of the enforcer, can obtain a ledger package that is complete in relation to the receipts (or assign blame to $f$$+$$1$ misbehaving \slow replicas). A complete ledger package contains all evidence that is necessary for the auditor to assign blame to misbehaving replicas if the receipts reflect any linearizability violation. In \S\ref{sec:incompatibility} and Lemma~\ref{lemma:incompatibility}, we show that, if a receipt does not appear correctly in a ledger package that is complete in relation to it, the auditor can assign blame to at least $f$$+$$1$ misbehaving replicas. In $\S$\ref{sec:violations} and Lemma~\ref{lemma:serializability-violations}, using the previous lemmas, we first prove that the auditor can assign blame correctly if it is given a set of receipts that reflects a serializability violation. Finally, Theorem~\ref{lemma:linearizability-violations} proves that, if a set of receipts reflects any linearizability violation, the auditor can assign blame to $f$$+$$1$ misbehaving \slow replicas.

\mypar{Minimum ledger index} Each client transaction request includes a field that specifies the minimum ledger index that it can be executed at. Correct replicas do not order a transaction~$t$ at ledger index~$i$, unless $i\geq m_i$ where $m_i$ is the minimum index value of $t$. Correct clients set the minimum index of a transaction to at least $M_i$$+$$1$ where $M_i$ is the largest value of the ledger index that they know of from the receipts that they have collected. The minimum index value is used to capture transaction dependencies efficiently and to reduce the amount of information that needs to be stored and transmitted to audit linearizability violations.

\mypar{Ledger well-formedness and validity} A ledger fragment is \emph{valid} if it can be produced by a sequence of correct primaries when there are at most $f$~misbehaving replicas.

A ledger fragment is \emph{well-formed} if either (i)~it is valid, or (ii)~it would be valid if not for the incorrect execution of one or more transactions, one or more incorrect checkpoint digests, or one or more invalid signatures or nonces.

A well-formed ledger matches the structural specifications of the \LPBFT protocol, \ie

\begin{itemize}

\item it specifies a serial ordering of transactions/entries, which respects their minimum ledger indices; and \item it includes evidence, and checkpoints at the required places.

\end{itemize}

A valid ledger is always well-formed, but a well-formed ledger can be invalid. A correct replica will never have a malformed ledger fragment, because replicas check the well-formedness of ledgers that they fetch. A correct replica may have an invalid ledger fragment. A ledger fragment can be well-formed but invalid only if there have at some point existed more than $N$$-$$f$$-$$1$ misbehaving replicas.

\mypar{Notation} Given a receipt $\langle \langle t_j, i_j, o_j \rangle, x_j \rangle$, we denote $\langle t_j, i_j, o_j \rangle$ by $\tio_j$. Unless explicitly defined otherwise, $s_j$ refers to the sequence number in $x_j$ of the receipt $\langle \tio_j, x_j \rangle$.

We say that a replica has ``signed a receipt'' if its signature is recorded in the receipt in the \msg{pre-prepare}/\msg{prepare} signatures' fields ($\sigma_p$ or in $\sum_s$).

\mypar{Receipt validity} A receipt is \emph{valid} if it is verifiable by \Cref{alg:receipt}.

\mypar{Preparement evidence for a batch} The \emph{preparement} evidence for a batch is $N$$-$$f$ signed \msg{pre-prepare}/\msg{prepare} messages for the batch, \ie $\mathcal{P}$ in \S\ref{sec:receipts}.

\mypar{Checkpoint sequence numbers} Let $\langle \tio_j, x_j \rangle$ be a valid receipt, $d_{C_j}$ be the checkpoint digest in $x_j$, and $C$ be the checkpoint interval. Anyone can calculate the sequence number at which the digest of the checkpoint is expected to be equal to $d_{C_j}$  as follows: checkpoints are always taken at sequence numbers that are multiples of $C$ and the digest in the receipt refers to the digest at the sequence number of the penultimate checkpoint transaction before $s_j$ (except the first $C$ transactions, which have the digest at genesis). So given $s_j$, the sequence number with the corresponding checkpoint digest, $s_\mathit{cp}$, can be calculated as

\begin{equation*}
  s_{\mathit{cp}} = 
  \begin{cases}
    0 &\text{if $s_j < C$}\\
    C \left( \lceil\frac{s_j}{C}\rceil-2 \right) &\text{otherwise.}
  \end{cases}
\end{equation*}

Note that the value of the digest itself is recorded in the last checkpoint transaction before $s_j$ (except the first $C$~transactions), \ie the checkpoint transaction that follows the one at $s_{\mathit{cp}}$. That checkpoint transaction is at

\begin{equation*}
  \begin{cases}
    0 &\text{if $s_j < C$}\\
    s_\mathit{cp} + C  &\text{otherwise.}
  \end{cases}
\end{equation*}

We assume that the genesis transaction~$\mathit{gt}$ is at sequence number~$0$.

\mypar{Fetching checkpoints}
Slow replicas can be brought up to date by fetching checkpoints and ledger fragments. When a correct replica fetches a checkpoint at sequence number $s$, it retrieves the ledger up to $s+C+P$. It first verifies the signatures in the evidence for the checkpoint transactions at $s$ and $s+P$. Note that the replicas that signed the checkpoint transaction at $s$ vouch for the validity of the ledger fragment between $s-C$ and $s$, whereas the replicas that signed the checkpoint transaction at $s+C$ vouch for the digest of the checkpoint at $s$. 

A correct replica, then, verifies that the digest of the checkpoint that it fetched matches the value recorded at $s+C$. It also checks, for each checkpoint transaction at sequence number $s'$ in the ledger, that the ledger's Merkle root at $s'$ matches the root in the evidence for the transaction at $s'$. Finally, the replica replays the ledger fragment between $s+1$ and $s+C$.

As noted previously, a correct replica may have a well-formed ledger fragment that includes invalid signatures as replicas do not verify all signatures in the ledger fragments that they fetch. Therefore, when contacted for an audit, a correct replica never returns a ledger fragment that it fetched with a checkpoint at sequence number $s$, without including the checkpoint transaction at $s+C$ and the evidence for that transaction.

\mysubsection{Obtaining the ledger}
\label{sec:obtain-complete-ledger}

\mypar{Ledger package} A \emph{ledger package} from a replica consists of one to four components:

\begin{enumerate}
    \item a ledger fragment~$\Fr$ that contains entries that locally prepared at the replica;
    \item an optional \emph{suffix}~$\Sf$ that contains entries that were preprepared atomically after a view-change but not yet prepared at the replica;
    \item an optional \emph{message box}~$\M$ that contains some of the messages from the replica's message box~$\mathcal{M}$; and 
    \item an optional \emph{checkpoint}~$\CP$.
\end{enumerate}

\mypar{Complete ledger package} Let $\R$ be a set of valid receipts; $s_{\max}$ be the maximum sequence number in $\R$; $s_{\min}$ be the sequence number of the checkpoint whose digest is expected to equal the checkpoint digest in the receipt with the smallest sequence number in $\R$ ($s_{\min}$ can be calculated as described in the previous section); $v_{\min}$ and $v_{\max}$ be the minimum and maximum view numbers in the receipts in $\R$, respectively.

A ledger package is \emph{complete} in relation to $\R$ if all of the following are true:

\begin{itemize}
    \item $\Fr + \Sf$ is well-formed;
    \item if $s_{\min} = 0$, $\CP$ contains the checkpoint at genesis (empty); otherwise, the digest of $\CP$ is equal to the one in the second checkpoint transaction in $\Fr+\Sf$;
    \item $\Fr$ includes at least one set of \vc and \nv messages for a view less than or equal to $v_{\min} + 1$ ($v_{\min}$ requirement), and one set of \vc and \nv messages for a view greater than or equal to $v_{\max}$ ($v_{\max}$ requirement);
    \item All signatures in $\Fr + \Sf$ and $\M$ are valid.
\end{itemize} 

and one of the following is true:

\begin{itemize}
    \item $\Fr$ includes entries between $s_{\min}$ and $s_{\max}+P$;
    \item $\Fr$ includes entries between $s_{\min}$ and  $s_{\max}+c$ where $c\in[0,P)$. $\M$ contains $P-c$ valid preparement evidence for entries from $s_{\max}-c$ to $s_{\max}$; or
    \item $\Fr$ includes entries between $s_{\min}$ and $e=\max(s_{\min}, s_{\max}-c)$ where $c\in[1,P]$. $\M$ contains valid preparement evidence for entries from $\max(s_{\min},e-P)$ to $e$.  
    The suffix $\Sf$ contains entries between $e+1$ and $s_{\max}$ that are preprepared but not prepared in some view $v' \geq v_{\max}$ and $\M$ contains preparement evidence from a view $<v'$ for entries between $e+1$ and $s_{\max}$.
\end{itemize}

\begin{lemma}[Obtaining a complete ledger package]\label{lemma:obtain-complete-ledger}
    Given a set of valid receipts $\R$, an auditor can either obtain a ledger package that is complete in relation to $\R$, or assign blame to at least $f+1$ misbehaving \slow replicas.
  \end{lemma}
  
\begin{proof}
  Select from the receipts in $\R$, the receipts with the highest view number $v_{\max}$. Then, from those receipts select the receipts with the highest sequence number. Finally, among those, let $R_{v\max}$ be the receipt with the highest index number. (We assume there is no tie; otherwise, the auditor assigns blame to the replicas that signed both tied receipts.)

    The enforcer asks all replicas that signed $R_{v\max}$ for a ledger package that is complete in relation to $\R$. We assume that correct replicas or members respond to the enforcer before the agreed deadline. Once the enforcer has responses from $f+1$ replicas, it relays the responses to the auditor; otherwise at the deadline, the enforcer assigns blame to at least $f+1$ misbehaving \slow replicas.
        
    We show that a correct replica can either respond with: a ledger package that is complete in relation to $\R$ or a ledger package with which the auditor can assign blame to $f+1$ misbehaving replicas. Therefore, after checking $f+1$ responses, the auditor either finds a complete ledger package, or assigns blame to $f+1$ misbehaving replicas.

    Note that a correct replica that is contacted by the enforcer can always satisfy the first three conditions of completeness: (1)~correct replicas always maintain well-formed ledgers and they record/can recalculate checkpoints; (2)~the $v_{\min}$ requirement can always trivially be satisfied by including the set of \vc and \nv messages for view $0$ in $\Fr$. In practice, for efficiency, correct replicas would satisfy this requirement by including the set of \vc and \nv messages for some view $v'$, where $v'$ is the latest possible in $[0, v_{\min}+1]$; and (3)~since the replicas that are asked are the replicas that signed $R_{v\max}$, they must have \vc and \nv messages for view $v_{\max}$. Therefore, any replica that returns a ledger package that violates any of the first three conditions can be assigned blame.

    The fourth condition of completeness requires that all signatures and the matching nonces in the ledger package are correct. Let $\langle \Fr, \Sf, \M, \CP \rangle$ be a ledger package returned by a replica. If $\Sf$ or $\M$ contains a message or transaction with an invalid signature, the auditor can assign blame to the replica. $\M$ contains messages from the replica's message box and $\Sf$ contains batches that pre-prepared at the replica. A correct replica never considers a message or pre-prepares a batch that includes an invalid signature. Otherwise, let $s_w$ be a sequence number where there is a transaction or message with an invalid signature. The auditor can look for the first checkpoint transaction that follows $s_w$ that has no invalid signatures in its evidence. If one exists, the auditor can assing blame to all $N-f$ replicas that signed that checkpoint transaction. If no such checkpoint transaction exists, the auditor can assign blame to the responding replica, since a correct replica never returns a ledger fragment that it has fetched with a checkpoint without including the committed checkpoint transaction that records that checkpoint's digest. So given a ledger package from a replica, the auditor can always verify all signatures and nonces in the package or assign blame to the responding replica or $N-f$ misbehaving replicas. So below, for brevity, we can assume that the ledger package that a replica returns has no invalid signatures or nonces.

    Additionally, for a correct replica that is contacted by the enforcer, one of the following must hold:

    \begin{itemize}  
        \item \textbf{The correct replica has locally prepared entries up to at least $s_{\max}$:} In this case, the replica can form a complete ledger package that includes either: 
        \begin{enumerate}[label=(\roman*)]
            \item a well-formed ledger fragment $\Fr$ that contains entries from $s_{\min}$ to $s_{\max} + P$; or
            \item a well-formed $\Fr$ that contains entries from $s_{\min}$ to $s_{\max}+c$ where $c\in[0,P)$, and $\M$ that contains $P-c$ valid preparement evidence for entries from $s_{\max}-c$ to $s_{\max}$.
        \end{enumerate}

        \item \textbf{The correct replica has not locally prepared entries up to $s_{\max}$ and it has locally prepared entries up to $e=s_{\max}-c$ where $c \geq 1$:} In this case, (1)~a correct replica can include entries between $s_{\min}$ and $e$ in a well-formed ledger fragment $\Fr$, and it can include the necessary preparement evidence in $\M$ (if $s_{\min}\leq e$); and (2)~if the replica has any batches that it has preprepared but not prepared due to a view-change, it can include the related \vc and \nv messages in $\Fr$ and the batches in $\Sf$. Let $p$ be the last sequence number for which there is a batch in $\Fr+\Sf$. If $p>e$, the correct replica can include the preparement evidence for entries between $e+1$ and $p$ in $\M$ as well. A correct replica can form a ledger package as described above. If $p \geq s_{\max}$, the ledger package is complete, and the replica can return it.  
    
        Otherwise, $p<s_{\max}$. Let $R_{s\max}$ be the receipt in $\R$ with the largest sequence number $s_{\max}$ and let $v_{s\max}$ be the view number in $R_{s\max}$. Note that $v_{s\max} \leq v_{\max}$ by definition, and in the correct replicas' ledger, there must exist at least one set of \vc and \nv messages for a view $v' > v_{s\max}$ such that none of the \vc messages include a \msg{pre-prepare} message for any batch at $s_{\max}$. The correct replica can return a ledger package that contains these \vc and \nv messages. The auditor can use the returned ledger package to assign blame to the intersection of replicas that signed $R_{s\max}$ and that sent the set of \vc messages for $v'$, as these replicas have prepared a batch at $s_{\max}$ but did not report it during the view change.

      \end{itemize} 

    Thus, for each of the $f+1$ responses, either the response is complete in relation to $\R$, or the auditor can assign blame to the misbehaving responder, or at least $f+1$ misbehaving replicas.
\end{proof}

By definition of completeness, if a ledger package is complete in relation to a set of valid receipts $\R$, it is complete in relation to any subset of $\R$. 

\mypar{Finding preparement evidence} For a batch at $s_r$, the auditor can find the preparement evidence for the batch as follows:

\begin{itemize}
    \item if $\Fr$ contains an entry at $s_r+P$, it is collected from there;
    \item if $\Fr$ contains the entry at $s_r$ but not at $s_r+P$, it is collected from $\M$; and
    \item if $\Fr$ does not contain an entry at $s_r$ but $\Sf$ contains an entry at $s_r$, it is also collected from $\M$, albeit it is for the same batch from a prior view. 
\end{itemize}

\mysubsection{Incompatibility}
\label{sec:incompatibility}

Let $R=\langle \tio_r,x_r \rangle$ be a valid receipt at sequence number $s_r$. Let $\langle \Fr, \Sf, \M, \CP \rangle$ be a ledger package that is complete in relation to $R$. Let $B_l$ be the batch that is at $s_r$ in $\Fr+\Sf$. $R$ is \emph{incompatible} with $B_l$ if any of the following hold:

\begin{itemize}
    \item $t_r$ does not appear in $B_l$;
    \item it does not appear in the $i_r$th position; or
    \item $o_r$ is different.
\end{itemize}

\begin{lemma}[Receipt-ledger incompatibility]\label{lemma:incompatibility}
    Let $R=\langle \tio_r,x_r \rangle$ be a valid transaction receipt for sequence number $s_r$. Let $\langle \Fr, \Sf, \M, \CP \rangle$ be a ledger package that is complete in relation to $R$. Let $B_l$ be the batch in the package at $s_r$. If $R$ is incompatible with $B_l$, the auditor can assign blame to at least $f+1$ misbehaving replicas.
\end{lemma}
\begin{proof}
    The auditor can calculate the set of replicas that signed $B_l$ using the preparement evidence that can be found as described above. These replicas are called $\E_l$. 
    
    Let $\E_r$ be the set of replicas that have signed the receipt. Let $v_r$ be the view number in the receipt and $v_l$ be the view number in the preparement evidence of $B_l$.

    \begin{itemize}  
        \item \textbf{$v_r = v_l$:} Correct replicas never sign  \msg{pre-prepare} or  \msg{prepare} messages for different batches in the same view. Therefore, the auditor can assign blame to the replicas in the intersection of $\E_r$ and $\E_l$, and $|\E_r \cap \E_l| \geq f + 1$.
        
        \item \textbf{$v_l > v_r$:} Correct replicas include the \msg{pre-prepare} messages for the last $P$ prepared batches in their \vc messages until the batches commit or a different batch is prepared at the sequence number. A correct primary always re-preprepares the latest batch that it finds in the set of $N-f$ \vc messages that it receives. Thus, there exists at least one view $v_{c} \in [v_{r}+1, v_l]$ where zero of the $N-f$ \vc messages for $v_c$ contain a \msg{pre-prepare} message for the batch at sequence number $s_r$ that is referenced in $R$. The ledger package is complete in relation to $R$, so $\Fr$ includes at least one set of \vc and \nv messages for a view less than or equal to $v_r+1$ (the $v_{\min}$ requirement). It must also include the set of \vc and \nv messages for $v_c$ as $v_l \geq v_c \geq v_r+1$.
        
        Let $\E_{c}$ be the set of replicas that have sent the \vc messages to the primary for view $v_c$. The auditor can assign blame to the replicas that are in the intersection of $\E_r$ and $\E_c$ and $|\E_r \cap \E_{c}| \geq f + 1$. 
        
        \item \textbf{$v_l < v_r$:} There exists at least one view $v_{c} \in [v_{l}+1, v_r]$ where zero of the $N-f$ \vc messages for $v_c$ contains a \msg{pre-prepare} message for the batch at sequence number $s_r$ that is referenced in $R$. The ledger package is complete in relation to $R$ so $\Fr$ includes at least one set of \vc and \nv messages for a view greater than or equal to $v_r$, so it must include the set of \vc and \nv messages for $v_c$ as $v_l+1 \leq v_c \leq v_r$ (the $v_{\max}$ requirement).
        Similar to previous case afterwards.
    \end{itemize}
\end{proof}

\mysubsection{Violations}
\label{sec:violations}

\mypar{Ordering receipts} Given a set of valid receipts, the auditor can order them lexicographically based on the corresponding (sequence number, index number, view number) tuples. (We can assume that there is no tie; otherwise, the auditor assigns blame to the replicas that signed both tied receipts.)
We say that a receipt $R_1$ is \emph{earlier/later} than a receipt $R_2$, if it is ordered before/after $R_2$ with this scheme, respectively. For example, the earliest receipt in a set of valid receipts is the one with the lowest view number, among those with the lowest index number, among those with the lowest sequence number.

\begin{lemma}[Serializability violations]\label{lemma:serializability-violations}
    Let $\R = \{(\tio_0, x_0), ..., (\tio_k, x_k))\}$ be a set of valid receipts that violates serializability. Then, the auditor can assign blame to at least $f+1$ misbehaving \slow replicas.
\end{lemma}
\begin{proof}
    
    First, the auditor can obtain a ledger package $\langle \Fr, \Sf, \CP, \M\rangle$ that is complete in relation to $\R$; otherwise, it can assign blame to at least $f+1$ misbehaving \slow replicas by Lemma~\ref{lemma:obtain-complete-ledger}. Note that, as the ledger package is complete in relation to $\R$, it is complete in relation to any receipt $R_j\in \R$.
 
    Since the receipts in $\R$ violate serializability, no serial execution of $t_0, ..., t_k$ can produce $io_0, ..., io_k$. $\Fr+\Sf$ is well-formed, so there are two options for its validity:
    
    \mypar{Valid ledger} $\Fr+\Sf$ is a valid ledger, so every transaction in it is ordered and executed serially. However, the receipts in $\R$ violate serializability. Therefore, there must exist at least one receipt $\langle \tio_w, x_w \rangle \in R$ that is incompatible with the batch at $s_w$ in $\Fr+\Sf$. By Lemma~\ref{lemma:incompatibility}, the auditor can assign blame to at least $f+1$ misbehaving replicas.
    
    \mypar{Invalid ledger} $\Fr+\Sf$ is a well-formed but invalid ledger. So there exists at least one transaction $t_w$ (which does not have to be in $\R$) that was executed incorrectly in some batch $s_w$, or one checkpoint that was created incorrectly.
        
    The auditor can order $\R$ as described above. Let $R_{e}$ be the earliest receipt in $\R$. Let $d_{C_0}$ be the checkpoint digest in $R_{e}$. Let $s_{C_0}$ be the sequence number with the expected checkpoint digest $d_{C_0}$, calculated by the auditor using $s_{e}$ and the checkpoint interval $C$ as previously described. If $s_{C_O} = 0$, but the digest in $R_e$ is not equal to the digest in the genesis transaction, the auditor can assign blame to all replicas that signed $R_e$. Otherwise, the ledger package is complete with respect to $R_{e}$, and $\Fr+\Sf$ is thus well-formed, so: (i)~the entry at $s_{C_0}$ in $\Fr+\Sf$ is a checkpoint transaction; and (ii)~the checkpoint transaction in $s_{C_O}+C$ exists as $s_{C_0} <s_{C_O}+C<s_{e}$ and contains the digest of $\CP$. If the digest of $\CP$ in the ledger package is not $d_{C_0}$, the auditor can assign blame to the replicas that signed both the checkpoint transaction at $s_{C_O}+C$ and $R_e$. The digest in that checkpoint transaction is for the previous checkpoint and the batches before the previous checkpoint have already committed since $C>P$.

    Otherwise, the auditor replays the ledger starting from the checkpoint transaction at $s_{C_0}$, creating checkpoints at checkpoint sequence numbers. Doing so, the auditor either obtains $\langle t_w, i_w, o_a \rangle \neq \langle t_w, i_w, o_w \rangle$ or finds that an incorrect checkpoint digest is recorded at $s_w$. In either case, the auditor can assign blame to all replicas that signed for the batch at $s_w$. 
\end{proof}

\begin{theorem}[Linearizability violations]\label{lemma:linearizability-violations}
    Let $\R$ be a set of receipts that violate linearizability. Then, the auditor can assign blame to at least $f+1$ misbehaving \slow replicas.
\end{theorem}

\begin{proof}
   If the receipts also violate serializability, the auditor can assign blame to at least $f+1$ misbehaving \slow replicas by Lemma~\ref{lemma:serializability-violations}. 

    Otherwise, since the receipts violate linearizability but not serializability, the ordering of the transactions in $\R$ must violate the real-time ordering of the transactions. So there exists at least two transactions, $t_a$ and $t_b$, in $\R$ such that the receipt for $\tio_a$ was received by the client before $t_b$ was sent, but $i_a \geq i_b$. $t_b$ was sent after $\langle \tio_a, x_a \rangle$ was received, so a correct client sets the minimum index $l$ of $\tio_b$ to at least $i_a+1$. Since $i_b \leq i_a$, the auditor can assign blame to all replicas who have sent the receipt for $\tio_b$.
\end{proof}

\mysection{Correctness of auditing with reconfiguration}
\label{sec:auditing-correctness-proof-reconfig}

In this section, we first summarize how reconfiguration happens, introduce new terminology, and update prior terminology. Then, in Lemma~\ref{lemma:governance-fork}, we prove that, if the auditor detects a fork in governance, it can assign blame to $f+1$ misbehaving replicas. In \S\ref{sec:obtain-complete-ledger-reconfig}, we update the prior discussion on obtaining a complete ledger package. In \S\ref{sec:config-mismatch} and Lemma~\ref{lemma:config-mismatch}, we prove that, if a receipt and the corresponding batch in a ledger package are prepared in different configurations, the auditor can assign blame to $f+1$ misbehaving replicas. In \S\ref{sec:incompatibility-reconfig}, using Lemma~\ref{lemma:config-mismatch}, we update the prior lemma about incompatibility. Finally, \S\ref{sec:violations-reconfig} updates the prior proofs on violations, and in Theorem~\ref{lemma:linearizability-violations-reconfig}, we prove the correctness of auditing in the complete \sys ledger system.

\mypar{Summary of reconfiguration} A correct primary ends the batch it is working on once it executes a governance transaction. Therefore, each batch includes at most one governance transaction and $i_g$ in a receipt refers to the last governance transaction executed before the transaction in the receipt. The final \vote transaction that is necessary to pass a referendum triggers the configuration change. $2P$~\emph{end-of-config} batches follow the final \vote before the configuration change. The governance sub-ledger consists of batches and evidence for all governance transactions. It also includes, for each configuration, the $P^{\mathrm{th}}$ and $2P^{\mathrm{th}}$ \emph{end-of-config} batches, which commit the final \vote transaction that triggers reconfiguration and the $P^{\mathrm{th}}$ \emph{end-of-config} batch respectively. The $P^{\mathrm{th}}$ \emph{end-of-config} batch links to the final \vote transaction, because its \msg{pre-prepare} message includes the Merkle root of the batch that includes the final \vote transaction.

\mypar{Updates to well-formedness and validity} A ledger fragment is \emph{valid} if it can be produced by a sequence of correct primaries in a sequence of configurations where in each configuration there are at most $f$ failures.

In addition to the previous structural specifications, governance changes are serialized and include the required \emph{end-of-config} and \emph{start-of-config} messages.

Note that correct replicas check the validity of the governance sub-ledger fragments that they fetch, so their governance sub-ledgers are valid, in addition to well-formed.

\mypar{Configuration number} The configuration number of a configuration $\C$ is the distance that it is from the configuration at the genesis. The genesis has configuration number~$0$. A configuration that follows the genesis configuration has number $1$ and so on.

\mypar{Supporting governance chain of a receipt} Every receipt~$R$ includes the index of the latest governance transaction. A correct client makes sure that it has a matching chain of valid governance transaction receipts for each receipt that it has. This includes the receipts for all governance transactions from the genesis up to the latest governance transaction, and the receipt for the $P^{\mathrm{th}}$ \emph{end-of-config}  batch for each configuration change. The supporting governance chain of a receipt~$R$ is the sequence of governance-related receipts that starts from the genesis transaction receipt and ends with the $P^{\mathrm{th}}$ \emph{end-of-config} batch receipt before the configuration that signed $R$ takes effect.

A supporting governance chain of a receipt matches a governance sub-ledger if each receipt in the chain is compatible with the governance sub-ledger. (For \emph{end-of-config} batches, compatibility considers committed Merkle roots as well.) Similarly, a supporting governance chain can be a prefix of a governance sub-ledger.

\mypar{Updates to receipt validity} A receipt is \emph{valid} if it is verifiable by \Cref{alg:receipt}, and it is attached a valid supporting governance chain.

\mypar{Updates to calculating checkpoint sequence numbers} If a sequence number that is multiple of the checkpoint interval~$C$ falls into an \emph{end-of-config}/\emph{start-of-config} sequence, checkpointing is skipped. A checkpoint is taken at the beginning of each new configuration, and the digest of the first checkpoint in a configuration is included in the first checkpoint transaction, as opposed to the one that follows (this is similar to genesis).

Let $\langle \tio_j, x_j \rangle$ be a valid receipt and $s_{\mathit{fv}}$ be sequence number of the final \vote transaction for the last configuration change in the supporting governance chain of the receipt. The first checkpoint of the configuration that prepared the receipt is expected at $s_{\mathit{fcp}} = s_{\mathit{fv}} + 2P + 1$. (Except the genesis configuration, for which $s_{\mathit{fcp}} = 0$.)

So given $s_j$, the sequence number $s_{\mathit{cp}}$ of the checkpoint whose digest is in $x_j$ can be calculated with

\begin{equation*}
  s_{\mathit{cp}} = 
  \begin{cases}
    s_{\mathit{fcp}} &\text{if $s_j < s_{\mathit{fcp}}+C$}\\
    C \left( \lceil\frac{s_j-s_{\mathit{fcp}}}{C}\rceil-2 \right) &\text{otherwise.}
  \end{cases}
\end{equation*}

\mypar{Updates to fetching checkpoints} Following a configuration change, a correct new replica fetches the checkpoint at the penultimate checkpoint sequence number $s'$ in the previous configuration (or the first checkpoint sequence number if there is only one). It also retrieves the full ledger. It replays the ledger from $s'$ before creating a checkpoint at the beginning of the configuration.

\mypar{Equivalence of $P^{\mathrm{th}}$ \emph{end-of-config} batches} Two $P^{\mathrm{th}}$ \emph{end-of-config} batches are equivalent if they:

\begin{enumerate}[label=(\roman*)]
    \item are at the same index and sequence number; and
    \item are preceded by the same valid governance sub-ledger (their \msg{pre-prepares} include the same committed Merkle root).
\end{enumerate}

Two receipts for $P^{\mathrm{th}}$ \emph{end-of-config} batches are equivalent if the batches specified in them are equivalent.

\mypar{Governance fork} There is a fork in governance if there is a fork in the governance sub-ledger. That is, there are at least two $P^{\mathrm{th}}$ \emph{end-of-config} batches for the same configuration number that belong in valid governance sub-ledgers, but that are not equivalent.

We say that there is a fork between two valid supporting governance chains if there are receipts for two $P^{\mathrm{th}}$ \emph{end-of-config} batches for the same configuration number that are not equivalent.

We say that there is a fork between a valid supporting governance chain and a valid governance sub-ledger, if for the same configuration number, the $P^{\mathrm{th}}$ \emph{end-of-config} batch specified by the receipt in the chain is not equivalent to the $P^{\mathrm{th}}$ \emph{end-of-config} batch in the sub-ledger.

\begin{lemma}[Governance fork]\label{lemma:governance-fork}
    If there is a fork in governance, the auditor can assign blame to at least $f+1$ misbehaving replicas.
\end{lemma}
\begin{proof}
    If there is a fork in governance, there are at least two $P^{\mathrm{th}}$ \emph{end-of-config} batches for the same configuration number that are not equivalent, namely $P_1$ and $P_2$.

    A correct replica only prepares a $P^{\mathrm{th}}$ \emph{end-of-config} batch at sequence number $s$ once the final \vote transaction that passes the referendum is committed at sequence number $s-P$. Thus, all governance transactions preceding it are committed too. This final \vote transaction triggers the configuration change.

    So the auditor can assign blame to the replicas that prepared both $P_1$ and $P_2$, because a correct replica that prepares one will never prepare another non-equivalent $P^{\mathrm{th}}$ \emph{end-of-config} batch in the same configuration number.
\end{proof}

\mypar{Longest supporting governance chain} Let $\R$ be a set of valid receipts. If there is a fork between the supporting governance chains of the receipts in $\R$, the auditor can assign blame to at least $f+1$ misbehaving replicas by Lemma~\ref{lemma:governance-fork}. So the auditor can always obtain a \emph{longest supporting governance chain} for the receipts in $\R$. This chain 
is the union of all supporting chains for receipts in $\R$.

Onwards, we assume that, given any set of valid receipts, the supporting governance chains are fork-free with each other and that there is a longest supporting governance chain; otherwise, the auditor can assign blame to $f+1$ misbehaving replicas by Lemma~\ref{lemma:governance-fork}.

\mypar{Transaction receipts} Onwards, we assume that a receipt is for a transaction and not for \emph{end-of-config}/\emph{start-of-config} batches. If the receipts for \emph{end-of-config}/\emph{start-of-config} indicate a fork in governance, misbehaving replicas can be blamed using Lemma~\ref{lemma:governance-fork}; otherwise, the \emph{end-of-config}/\emph{start-of-config} batches do not have any usage and do not affect the key-value store, so do not affect linearizability.

\mysubsection{Updates to obtaining the ledger}
\label{sec:obtain-complete-ledger-reconfig}

\mypar{Updated ledger package} A ledger package includes an additional required field:
\begin{itemize}
    \item the committed governance sub-ledger $\GL$ of the replica. 
\end{itemize}

\mypar{Updated definition of completeness} Let $\R$ be a set of valid receipts. Define $s_{\max}, v_{\min},v_{\max}$ as previously. Calculate $s_{\min}$ using the receipt with the smallest configuration number, among those with the smallest sequence number in $\R$. Let $\G_{g\max}$ be the longest supporting governance chain in $\R$.

A ledger package is \emph{complete} in relation to $\R$ if, in addition to the prior conditions about well-formedness, length, and $v_{\min}$/$v_{\max}$ requirements:
\begin{itemize}
    \item $\G_{g\max}$ is a prefix of $\GL$ (i.e. the package is obtained from a replica in a configuration which is equal to or succeeds all configurations in $\R$);
    \item $\GL$ is valid; and
    \item $\GL$ matches $\Fr$.
\end{itemize} 

The condition for the checkpoint $\CP$ is updated as follows:

\begin{itemize}
\item if $s_{\min}$ is calculated as the first checkpoint transaction in a configuration (or zero), the digest of $\CP$ is equal to the one in the checkpoint transaction at $s_{\min}$; otherwise, the digest of $\CP$ is equal to the one in the second checkpoint transaction in $\Fr+\Sf$.
\end{itemize} 

\begin{lemma}[Obtaining a complete ledger package with reconfiguration]\label{lemma:obtain-complete-ledger-reconfig}
    Given a set of valid receipts $\R$, an auditor can either obtain a ledger package that is complete in relation to $\R$, or assign blame to at least $f+1$ misbehaving \slow replicas.
  \end{lemma}
  
\begin{proof}
    As mentioned before, we assume that there is no fork between the supporting governance chains of the receipts in $\R$. Let $R_{g\max}$ be the receipt with the highest index number, among those with the highest sequence number, among those with the highest view number, among those with the longest supporting governance chain in $\R$. Let $\G_{g\max}$ be the supporting governance chain of $R_{g\max}$.

    We assume that there is a reliable mechanism to look up the most recent system configuration. Using this mechanism, the auditor looks up the most recent committed governance sub-ledger and the set of replicas that signed the first checkpoint transaction of the most recent configuration. If there is a fork between $\G_{g\max}$ and the governance sub-ledger that is looked-up, the auditor can assign blame to at least $f+1$ misbehaving replicas by Lemma~\ref{lemma:governance-fork}; otherwise, the auditor checks whether the sub-ledger that is looked up is longer than $\G_{g\max}$. If so, the enforcer asks all the replicas that signed the first checkpoint transaction of the most recent configuration for a ledger package; otherwise, the replicas that have signed $R_{g\max}$ are asked.
        
    As in Lemma~\ref{lemma:obtain-complete-ledger}, the enforcer asks replicas for a ledger package that is complete in relation to $\R$. At the deadline, the enforcer relays the responses to the auditor. There are at least $f+1$ responses, or the enforcer can assign blame to $f+1$ misbehaving \slow replicas. 

    As before, we show that a correct replica can either respond with: a ledger package that is complete in relation to $\R$, or a ledger package with which the auditor can assign blame to $f+1$ misbehaving replicas.

    First, note that a correct replica that is contacted by the enforcer can always satisfy the updated  completeness conditions (related to $\GL$), because the replica is part of the most recent configuration and the conditions all pertain to keeping a valid governance sub-ledger. Of the conditions described previously, the well-formedness and $v_{\min}$ conditions can be satisfied, and invalid signatures in the package can be handled, just as in Lemma~\ref{lemma:obtain-complete-ledger}. Since the replicas that are asked are not necessarily the replicas that signed the receipt with the highest view in $\R$, it is possible that they cannot satisfy the $v_{\max}$ requirement even if they are correct. 

    So, for a correct replica that is contacted by the enforcer one of the following must hold:

    \begin{itemize}
        \item \textbf{The replica cannot satisfy the $v_{\max}$ requirement:}
        Let $R_{v\max}$ be the latest receipt when the receipts are ordered lexycographically by 
        (view number, configuration number, sequence number, index number). Let $\G_{v\max}$ be the supporting governance chain of $R_{v\max}$. If there is a fork between $\G_{v\max}$ and the committed sub-ledger $\GL$ of the replica, the replica can return its governance sub-ledger and the auditor can assign blame to at least $f+1$ misbehaving replicas by Lemma~\ref{lemma:governance-fork}. Otherwise, $\G_{v\max}$ must be a prefix of $\GL$ since the enforcer asks replicas from the most recent configuration. There are two possibilities for the relationship between $\G_{v\max}$ and $\GL$:

        \begin{enumerate}
            \item $\GL = \G_{v\max}$. So $R_{g\max} = R_{v\max}$.Therefore, the correct replica signed $R_{v\max}$. 
            Any correct replica that signed $R_{v\max}$ has the \vc and \nv messages for $v_{\max}$, so this case is a contradiction.
            \item $\GL$ is longer than $\G_{v\max}$. Let $P_{v\max+1}$ be the $P^{\mathrm{th}}$ \emph{end-of-config} batch that ends $R_{v\max}$'s configuration $C$. Since the replica is correct and cannot satisfy the $v_{\max}$ requirement, $P_{v\max+1}$ must be prepared in a view $<v_{\max}$. Any correct replica that prepared $P_{v\max+1}$ must have committed a final \vote transaction that triggers the configuration change in their ledger in a view less than $v_{\max}$. Since correct replicas never reset their ledger by more than $P$ sequence numbers, they do not pre-prepare any batch with view $v_{\max}$ in $C$.  So, the auditor can assign blame to the intersection of replicas that signed $R_{v\max}$ and prepared $P_{v\max + 1}$.
        \end{enumerate}

        \item \textbf{The replica can satisfy the $v_{\max}$ requirement:}
        If, additionally, the replica has prepared (or pre-prepared with view changes) batches up to at least $s_{\max}$, it can return a ledger package that is complete in relation to $\R$ just as in Lemma~\ref{lemma:obtain-complete-ledger}.

        Otherwise, let $R_{s\max}$ be the receipt with the largest sequence number $s_{\max}$. Let $\G_{s\max}$ be the supporting governance chain of $R_{s\max}$. If there is a fork between $\G_{s\max}$ and the replica's $\GL$, the replica can return $\GL$ and the auditor can assign blame to at least $f+1$ misbehaving replicas by Lemma~\ref{lemma:governance-fork}. Otherwise, $\G_{s\max}$ must be a prefix of $\GL$ since the replicas asked by the enforcer are from the most recent configuration. Again, there are two possibilities:
        \begin{enumerate}
            \item \textbf{$\GL$ is longer than $\G_{s\max}$:} Let $P_{s\max+1}$ be the $P^{\mathrm{th}}$ \emph{end-of-config} batch that ends $R_{s\max}$'s configuration. Since the replica is correct and cannot satisfy the $s_{\max}$ requirement, $P_{s\max+1}$ must be prepared at a sequence number less than $s_{\max}$. Any correct replica that prepared $P_{s\max+1}$ must have committed a final \vote transaction that triggers the configuration change at latest at sequence number $s_{\max} - (P+1)$. Since a correct replica never resets its ledger by more than $P$ sequence numbers, the auditor can assign blame to the replicas that signed both $R_{s\max}$ and prepared $P_{s\max + 1}$.
            
            \item \textbf{$\GL = \G_{s\max}$:} The group of replicas asked by the enforcer are from the same configuration that signed $R_{s\max}$, which is the most recent configuration. 
            Since the replica is correct and from the most recent configuration $v_{s\max} \leq v_{\max}$ by definition. In $\Fr$, as before, there must exist at least one set of \vc and \nv messages for a view $v' > v_{s\max}$ such that none of the \vc messages includes a \msg{pre-prepare} for any batch at $s_{\max}$. Note that the configuration of the replicas that have sent these \vc messages must be the same as the configuration that signed the receipt, as that is the most recent configuration in the system. So just as in Lemma~\ref{lemma:obtain-complete-ledger}, the auditor can assign blame to the replicas that signed both $R_{s\max}$ and that sent the set of \vc messages for $v'$.
        \end{enumerate}
    \end{itemize} 
    
    So, for each of the $f+1$ responses, either the response is complete in relation to $\R$, or the auditor can assign blame to the responder, or at least $f+1$ misbehaving replicas.
\end{proof}

\mysubsection{Mismatching configurations}
\label{sec:config-mismatch}

\begin{lemma}[Receipt-ledger configuration mismatch]\label{lemma:config-mismatch}
    Let $R = \langle \tio_r, x_r \rangle$ be a valid receipt that was produced in a configuration $\C_r$. Let $B_l$ be the batch that is at $s_r$ in a ledger package that is complete in relation to $R$. Let $\C_l$ be the configuration of the replicas that signed $B_l$. If $\C_r \neq \C_l$, the auditor can assign blame to at least $f+1$ misbehaving replicas.
\end{lemma}
\begin{proof}
    Since $R$ is a valid receipt, it has a valid supporting governance chain. Since the ledger package is complete, it includes a valid governance sub-ledger $\GL$ that leads to $\C_l$, which is fork-free with the supporting governance chain of $R$.

    One of the following must hold:
    \begin{itemize}
        \item \textbf{$\C_r < \C_l$: $\C_r$ precedes $\C_l$:} Let $P_{r+1}$ be the $P^{\mathrm{th}}$ \emph{end-of-config} batch that ends the configuration $\C_r$. This batch and its evidence is included in $\GL$. Since the package is complete, $\GL$ is consistent with the ledger fragment in the package. Since that ledger fragment is well-formed and $B_l$ is at $s_r$, $P_{r+1}$ is at the latest at sequence number $s_r-(P+1)$. Any replica that prepared $P_{r+1}$ must have committed a final \vote transaction that triggers the configuration change at the latest at sequence number $s_r-(2P+1)$. A correct replica that has prepared a batch at $s_r$ in $\C_r$ never resets its ledger to earlier than $s_r-P$ even with view changes. So the auditor can assign blame to the replicas that both prepared $P_{r+1}$ and signed $R$.
        
        \item \textbf{$\C_r > \C_l$: $\C_r$ succeeds $\C_l$:} We show that this case is impossible given that $R$ is valid, and there is no fork between its supporting governance chain and $\GL$. Since the ledger package is complete in relation to $R$, $\GL$ includes the $P^{\mathrm{th}}$ \emph{end-of-config} batch leading to $\C_r$ and it matches the well-formed ledger fragment in the package. Since $B_l$ is at $s_r$, that batch can at earliest be at sequence number $s_r+P$. So there cannot be a valid receipt produced in $\C_r$ at $s_r$.
    \end{itemize}
\end{proof}

\mysubsection{Updates to incompatibility}
\label{sec:incompatibility-reconfig}

\begin{lemma}[Receipt-ledger incompatibility with reconfiguration]\label{lemma:incompatibility-reconfig}
    Let $R=\langle \tio_r,x_r \rangle$ be a valid transaction receipt at sequence number $s_r$. Let $\langle \Fr, \Sf, \M, \CP, \GL \rangle$ be a ledger package that is complete in relation to $R$. Let $B_l$ be the batch in the package at $s_r$. If $R$ is incompatible with $B_l$, the auditor can assign blame to at least $f+1$ misbehaving replicas.
\end{lemma}
\begin{proof}
    Define $\E_l, \E_r, v_l, v_r$ as in Lemma~\ref{lemma:incompatibility}. Note that we can assume that both the receipt and $B_l$ are prepared by the same configuration $\C$; if not, the auditor can assign blame to $f+1$ misbehaving replicas by Lemma~\ref{lemma:config-mismatch}.

    \begin{itemize}  
        \item \textbf{$v_r = v_l$:} Same as Lemma~\ref{lemma:incompatibility}.
        
        \item \textbf{$v_l > v_r$:} Calculate $\E_{c}$ as described in Lemma~\ref{lemma:incompatibility}. If the replicas in $\E_c$ are also from the configuration $\C$, the auditor can assign blame just as in Lemma~\ref{lemma:incompatibility}; otherwise, if the replicas in $\E_c$ are from a preceding configuration, the first checkpoint transaction of $\C$ is at the latest at sequence number $s_r-(P+1)$ since $B_l$ is prepared by $\C$ and $\Fr+\Sf$ is well-formed. Furthermore, that checkpoint transaction is prepared in a view $v' > v_r$. A correct replica never signs the receipt at $s_r$ in a view $v_r$ and then resets its ledger by more than $P$ sequence numbers while view changing to $v'$. So, the auditor can assign blame to the replicas that signed both that checkpoint transaction and the receipt. 
        
        \item \textbf{$v_l < v_r$:}
        Calculate $\E_{c}$ as described in Lemma~\ref{lemma:incompatibility}.
        If the replicas in $\E_c$ are also from the configuration $\C$, the auditor can assign blame just as in Lemma~\ref{lemma:incompatibility}; otherwise the replicas in $\E_c$ are from a configuration that succeeds $\C$. In this case, the $P^{\mathrm{th}}$ \emph{end-of-config} batch that ends the configuration $\C$ is at the earliest at sequence number $s_r+P$, since $B_l$ is prepared by $\C$ and $\Fr+\Sf$ is well-formed. Furthermore, that batch is prepared in a view $v' < v_r$. A correct replica that prepares that $P^{\mathrm{th}}$ \emph{end-of-config} batch commits to the configuration change; it never resets its ledger to earlier than $s_r$ and signs $R$. So, the auditor can assign blame to the replicas that signed both that \emph{end-of-config} batch and the receipt.
    \end{itemize}
\end{proof}

\mysubsection{Updates to violations}
\label{sec:violations-reconfig}

\begin{lemma}[Serializability violations with reconfiguration]\label{lemma:serializability-violations-reconfig}
    Let $\R = \{(\tio_0, x_0), ..., (\tio_k, x_k))\}$ be a set of receipts that violates serializability. Then, the auditor can assign blame to at least $f+1$ misbehaving \slow replicas.
\end{lemma}
\begin{proof}
    First, the auditor can obtain a ledger package $\langle \Fr, \Sf, \CP, \M, \GL \rangle$ that is complete in relation to $\R$; otherwise, \sys can assign blame to at least $f+1$ misbehaving \slow replicas by Lemma~\ref{lemma:obtain-complete-ledger-reconfig}.
    
    Just as in Lemma~\ref{lemma:serializability-violations}, since the receipts in $\R$ violate serializability, no serial execution of $t_0, ..., t_k$ can produce $io_0, ..., io_k$. $\Fr$ is well-formed, so there are two options for its validity:
    
    \mypar{Valid ledger} Similar to Lemma~\ref{lemma:serializability-violations}. By Lemma~\ref{lemma:incompatibility-reconfig}, the auditor can assign blame to at least $f+1$ misbehaving replicas.
     
    \mypar{Invalid ledger} Assume that receipts are ordered lexicographically based on the corresponding (sequence number, configuration number, index number, view number) tuples. (We can assume that there is no tie; otherwise the auditor can assign blame to the replicas that signed both tied receipts.)
    
    Let $R_{e}$ be the earliest receipt in the ordered $\R$. Let $d_{C_0}$ be the digest in $R_{e}$. Let $s_{C_0}$ be the sequence number with the expected checkpoint digest $d_{C_0}$. $s_{C_0}$ can be calculated by the auditor using $s_{e}$, the checkpoint interval $C$, and the supporting governance chain. (Note that $s_{C_0}$ is equal to $s_{\min}$ that is calculated while obtaining the ledger.)
    
    We can assume that the batch at $s_e$ is prepared by the same configuration that sent the receipt; otherwise the auditor can assign blame to $f+1$ misbehaving replicas by Lemma~\ref{lemma:config-mismatch}. We also know that the supporting governance chain of $R_e$ matches $\Fr+\Sf$ and that $\Fr+\Sf$ is well-formed. So, the checkpoint transactions at $s_{C_0}$ (and $s_{C_0} + C$ if it exists) are prepared by the same configuration as $R_e$ by definition of $s_{C_0}$. So, if the digest at $s_{C_0}$ is not $d_{C_0}$, the auditor can assign blame to $f+1$ misbehaving replicas similar to Lemma~\ref{lemma:serializability-violations}. 

    Since the supporting governance chains of all receipts match the ledger fragment by definition of completeness, the auditor can determine the correct stored procedures for each transaction to replay the ledger as in Lemma~\ref{lemma:serializability-violations}.
\end{proof}

\begin{theorem}[Linearizability violations with reconfiguration]\label{lemma:linearizability-violations-reconfig}
    Let $\R$ be a set of receipts that violate linearizability. Then, the auditor can assign blame to at least $f+1$ misbehaving \slow replicas.
\end{theorem}
\begin{proof}
    If the receipts also violate serializability, the auditor can assign blame to at least $f+1$ misbehaving \slow replicas by Lemma~\ref{lemma:serializability-violations-reconfig}; otherwise, the minimum ledger index argument in the proof of Theorem~\ref{lemma:linearizability-violations} holds.
\end{proof}


\end{document}